\documentclass[a4paper,UKenglish]{lipics-v2018}

\usepackage{amssymb,amsmath,latexsym,wasysym,dsfont,stmaryrd}
\usepackage{amsthm}
\usepackage{soul}
\usepackage{caption}
\usepackage{bm} 
\usepackage{xparse} 
\usepackage{xcolor}
\usepackage{pgf,tikz}
\usetikzlibrary{positioning,automata}
\usetikzlibrary{shapes.multipart}
\usepackage{algorithm}
\usepackage{algpseudocode}
\usepackage{wrapfig}

\usepackage{xspace,microtype,hyperref}

\usepackage{macros,environments,colored-comments}

\title{Changing Observations in Epistemic Temporal Logic}
\titlerunning{Changing Observations in Epistemic Temporal Logic}

\author{Aur\`ele Barri\`ere}
{ENS Rennes, France}
{aurele.barriere@ens-rennes.fr}
{}
{}
\author{Bastien Maubert}
{University of Naples ``Federico II'', Naples, Italy}
{bastien.maubert@gmail.com}
{}
{This project has received funding from the European Union's Horizon 2020 research
  and innovation programme under the Marie Sklodowska-Curie grant agreement No 709188.}
\author{Aniello Murano}
{University of Naples ``Federico II'', Naples, Italy}
{murano@na.infn.it}
{}
{}
\author{Sasha Rubin}
{University of Naples ``Federico II'', Naples, Italy}
{sasha.rubin@unina.it}
{}
{}

\authorrunning{A. Barri\`ere, B. Maubert, N. Murano and S. Rubin} 

\Copyright{Aur\`ele Barri\`ere, Bastien Maubert, Aniello Murano and
Sasha Rubin}

\subjclass{Logic and Verification}

\EventEditors{John Q. Open and Joan R. Access}
\EventNoEds{2}
\EventLongTitle{42nd Conference on Very Important Topics (CVIT 2016)}
\EventShortTitle{CVIT 2016}
\EventAcronym{CVIT}
\EventYear{2016}
\EventDate{December 24--27, 2016}
\EventLocation{Little Whinging, United Kingdom}
\EventLogo{}
\SeriesVolume{42}
\ArticleNo{23}
\nolinenumbers 
\hideLIPIcs  

\newcommand\UElogo{%
\begin{tikzpicture}[remember picture,overlay]
\node[anchor=south,yshift=4.2cm,xshift=2cm] at (current page.south) {\includegraphics[height=2.5em]{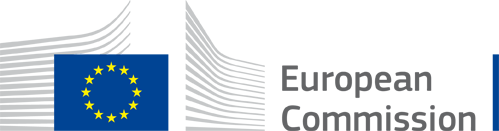}};
\end{tikzpicture}%
}

\begin{document}


\maketitle
\UElogo

\begin{abstract}
We study  dynamic changes of agents' observational power
in  logics of knowledge and time. We  consider \CTLsK, the extension
of \CTLs with knowledge operators, and enrich it with a new operator
that models a change in an agent's  way of observing the system.
We extend the classic semantics of knowledge for perfect-recall agents
to account for changes of observation, and we show that this new operator strictly increases
the expressivity of \CTLsK. We  reduce the model-checking problem for
our logic to
that for \CTLsK, which is known to be decidable. This provides a
solution to the model-checking problem for our logic, but its complexity
is not optimal. Indeed we  provide a direct
decision procedure with better complexity. 


\end{abstract}


\section{Introduction}
\label{sec:intro}
In multi-agent systems,  agents usually have only partial
information about the state of the
system~\cite{van2003cooperation}. This has led to the development of
epistemic logics, often combined with temporal logics, for describing
and reasoning about how distributed systems and agents' knowledge
evolve over time. Such formalisms have been applied to the modelling
and analysis of, \eg, distributed
protocols~\cite{DBLP:conf/tark/LadnerR86,RAK}, information flow and
cryptographic
protocols~\cite{DBLP:conf/csfw/MeydenS04,DBLP:journals/jcs/HalpernO05} and
knowledge-based programs~\cite{van1998synthesis}.


In these frameworks, an agent's view of a particular state of the
system is given by an {observation} of that state. In all the
cited settings, an agent's observation of a given state does not 
change over time. In other words, these frameworks have no primitive
for reasoning about agents whose observation power can change. \bmchanged{Because
this phenomenon occurs in real scenarios, for instance when a
user of a system is granted access to previously hidden data, we
propose here to tackle this problem.}
Precisely, 
we extend classic epistemic temporal logics with a new
unary operator, $\D{o}$, that represents changes of observation power, and is read ``the agent changes her observation power to
$o$''. 
For instance, the formula $\D{o_1}\A\F(\D{o_2}(\K p\vee\K\neg p))$
expresses that ``For an agent with initial observation power $o_1$,
in all possible futures there exists a point where, if the agent  updates
her observation power to $o_2$, she  learns whether or not the
proposition $p$ holds''. If in this example $o_1$ and $o_2$ represent
different ``security levels'' and $p$ is sensitive information, then
the formula expresses a possible avenue for attack.
The logics and model-checking procedures that we present in this paper
allow the expression and evaluation of such properties. 
Another motivation for studying such logics comes from the recently introduced
 Strategy Logic with Imperfect
Information~\cite{DBLP:conf/lics/BerthonMMRV17}, an extension of
Strategy Logic~\cite{DBLP:journals/tocl/MogaveroMPV14} in which
agents can dynamically change observation power when 
changing strategies.


There is a rich history of epistemic logic in AI, including the static
and temporal settings~\cite{RAK}, the dynamic setting~\cite{ditmarsch2007dynamic} as well as the
strategic setting~\cite{van2003cooperation}. 
%
The most common logics of knowledge and time are \CTLK, \LTLK and
\CTLsK, which extend the classic temporal logics \CTL, \LTL and \CTLs
with epistemic operators. 
These logics have been studied for two main recall abilities: 
no memory or  perfect recall.   
\bmchanged{For memoryless agents, adapting the semantics of these logics  to
include the observation-change operator is straightforward.
Model-checking algorithms also can be easily adapted 
by keeping updated in the procedure 
each agent's current observation, and this
information is of linear size.
The resulting logics thus have a \PSPACE-complete model-checking
problem, as  \LTLK, \CTLK and \CTLsK
do}~\cite{raimondi2005complexity,DBLP:conf/atal/KongL17}. 

The case of agents with perfect recall, which we study in this work,
is more delicate. The model-checking problem for \LTLK and \CTLsK is
nonelementary
decidable~\cite{DBLP:conf/fsttcs/MeydenS99,BOZZELLI201580,DBLP:conf/atal/Aucher14},
with \kEXPTIME
upper-bound for formulas with at most $k$ nested knowledge operators. 
The same upper-bounds are known for \CTLK~\cite{Dima2009}.
In this work we show that, as for the memoryless semantics, the
introduction of observation changes in epistemic temporal logics with
perfect recall  does not increase the complexity
of the model-checking problem.

We extend \CTLsK (which subsumes \CTLK and \LTLK) with
observation-change operators $\D{o}$, starting with the 
single-agent case, and we solve its model-checking problem.  \bmchanged{To do so
we define an alternative
semantics which, unlike the natural one, is
 based on a bounded amount of information. Once the two semantics proven
 to be equivalent, designing a
 model-checking algorithm is almost straightforward.} 
We then extend the logic to the multi-agent case,  introducing
 operators $\D{o}_a$ for each agent $a$,
 and we extend our approach to solve its 
model-checking problem. Next, we study the expressivity of our logic,
showing that the observation-change operator increases
expressivity. We finally provide a reduction to \CTLsK which shows how
to remove observation-change operators, at the cost of a blow-up in
the size of the model. Our direct model-checking procedure is shown to have a
better complexity than going through this reduction and using known
 model-checking algorithms   for \ctlsk.




\section{\ctlskd}
\label{sec:ctlskd}

In this section we define the logic \ctlskd, which  extends \CTLsK.
We first
 study the case where there is only one agent (and thus only one
knowledge operator). We will extend to the multi-agent setting in Section~\ref{sec:multi}.

\subsection{Notation}
A \emph{finite} (resp. \emph{infinite}) \emph{word} over some alphabet $\Sigma$ is an element
of $\Sigma^{*}$ (resp. $\Sigma^{\omega}$). 
The \emph{length} of a finite word $w=w_{0}\ldots
w_{n}$ is $|w|\egdef n+1$, and we let $\last(w)\egdef w_{n}$.
Given a finite (resp. infinite) word $w$ and $0 \leq i \leq |w|$ (resp. $i\in\setn$), we let $w_{i}$ be the
letter at position $i$ in $w$, $w_{\leq i}$ is the prefix of $w$ that
 ends at position $i$,
and $w_{\geq i}$ is the suffix  that starts
 at position $i$.
 We write $w\pref w'$ if $w$ is a prefix of $w'$. 

\subsection{Syntax}


We fix  a countably infinite set of atomic
propositions,  $\AP$. We also let $\Obs$ be a finite set of \textit{observations},
that  represent  possible observational powers of the agent.
Note that in this work,  
``observation'' does not refer to a punctual
observation of a  system's state, but rather a way of observing
the system, or ``observational power'' of an agent. 

As for state and path
formulas in \CTLs, we distinguish between \textit{history formulas} and \textit{path
  formulas}. We say history formulas
instead of state formulas because, considering agents with \emph{perfect recall} of
the past, the truth of epistemic
formulas depends not only on the current state, but also on the history
before reaching this state. 

\begin{definition}[Syntax]
The sets of history formulas $\phi$ and path formulas $\psi$ are
defined by the following grammar:
\[
\begin{array}{ccl}
  \phi & ::= & p ~|~ \neg \varphi ~|~ \varphi\wedge\varphi ~|~ \A\psi
               ~|~ \K\varphi ~|~ \D{o}\varphi\\
  \psi & ::= & \varphi ~|~ \neg\psi ~|~ \psi\wedge\psi ~|~ \X\psi ~|~ \psi\U\psi,
\end{array}
\]
where $p\in\AP$ and $o\in\Obs$. 
\end{definition}

We call \ctlskd formulas  all
history formulas so defined.
Operators $\X$ and $\U$ are the
classic \textit{next} and \textit{until} operators of temporal logics, and
$\A$ is the universal path quantifier from branching-time temporal
logics. 
$\K$ is the knowledge operator from
epistemic logics, and
 $\K\varphi$ reads as ``the agent knows that $\varphi$ is true''.
Our new \emph{observation change} operator, $\D{o}$,
reads as ``the agent now observes the system with observation $o$''. 

As usual, we define $\top=p\vee\neg p$, $\phi\vee\phi'=\neg(\phi\wedge\neg\phi)$,
$\phi\impl\phi' = \neg \phi \vee \phi'$, as well as the temporal operators \emph{finally} ($\F$) and
\emph{always} ($\G$): $\F\phi = \top\U\phi$, and $\G\phi=\neg\F\neg\phi$.

\subsection{Semantics}

The models  of \ctlskd are  Kripke
structures equipped with one observation relation $\eqstate{o}$ on states for
each observation $o$.  

\begin{definition}[Models]
  \label{def-models}
A \emph{Kripke
structure with observations} is given as a structure
$M=(\APf,S,T,V,\{\eqstate{o}\}_{o\in\Obs},\sinit,\oinit)$, where
\begin{itemize}
\item  $\APf\subset \AP$  is a finite subset of atomic propositions,
\item $S$ is a set of states,
\item $T\subseteq S\times S$ is a left-total\footnote{\ie, for every $s\in
  S$ there exists $s'\in S$ such that $s T s'$. This cosmetic
  restriction is made to avoid having to deal with finite runs ending
  in deadlocks.
} transition relation,
\item $V:S\rightarrow 2^{\APf}$ is a valuation function, 
\item $\eqstate{o}\subseteq S\times S$ is an equivalence relation, for each
  $o\in\Obs$,
  \item $\sinit\subseteq S$ is an initial state, and
\item $\oinit\in\Obs$ is the initial observation.
\end{itemize}  
\end{definition}

A \textit{path} is an infinite sequence of states
$\pi=s_0 s_1\dots$ such
that  for all $i\geq 0$, $s_i T s_{i+1}$, and a \textit{history}
$h$ is a finite prefix of a path.
For $I\subseteq S$, we write
$T(I)=\{s' \mid \exists s\in I \mbox{ s.t. }s T s'\}$ for the set of
successors of states in $I$. Finally, for $o\in\Obs$ and $s\in S$, we 
let $\eqc{s}=\{s' \mid s \eqstate{o} s'\}$ be the equivalence class
of $s$ for relation $\eqstate{o}$.

\begin{remark}
  \label{rem-model}
  We model agents' information via indistinguishability
  relations $\eqstate{o}$, where $s\eqstate{o}s'$ means that $s$ and
  $s'$ are indistinguishable for an agent who has observation power $o$.
  Other approaches exist. 
  One is via  \emph{observation functions}
  (see, \eg, \cite{DBLP:conf/fsttcs/MeydenS99}), that
  map states to atomic observations, and where two states are
  indistinguishable for an observation function if they have the same
  image. Another consists in seeing states as tuples of \emph{local
    states}, one for each agent, two global states being indistinguishable
  for an agent if her local state is the same in both (see, \eg,
  \cite{DBLP:conf/atal/KongL17}).
  All these formalisms are essentially equivalent with respect to
  epistemic temporal logics~\cite{pacuit2007some}. 
  In these alternative formalisms, observation change would correspond
  to, respectively, changing
  observation function, and changing local states in global
  states. Note that the formalism via local states, because
  indistinguishability is hard-coded in states of the system, is less
  adapted to
  the modelling of observation change.
\end{remark}

\vspace{1ex}
\head{Observation records}
To define which histories the agent cannot distinguish, we need to
keep track of how she observed the system at each point in time. To do
so, we record each observation change as a pair $(o,n)$, where $o$ is
the new observation and $n$ is the time when this change occurs.

\begin{definition}
An  \emph{observation record} is a finite word over $\Obs\times\setn$.
\end{definition}

Note that observation records are meant to represent changes of
observational ability, and thus they do not contain the initial
observation (which is given in the model). We write $\eobsrec$ for
the empty observation record.

\begin{example}
  \label{ex-1}
Consider a  model $M$ with initial observation $\oinit$, a history
$h=s_0\ldots s_4$ and an observation record
$r=(o_1,0)\cdot(o_2,3)\cdot(o_3,3)$. The agent first observes state $s_0$
with observation $\oinit$. The observation record shows that at time
0, thus before the first transition, the agent changed for observation
$o_1$. She then observed state $s_0$ again, but this time with
observation $o_1$. Then the system goes through states $s_1$ and
$s_2$ and reaches $s_3$, all of which she observes  with observation $o_1$. 
At time 3, the agent changes to observation $o_2$, and thus observes
state $s_3$ again, but this time with  observation $o_2$, and finally
she switches to observation $o_3$ and thus observes $s_3$ once more,
with observation $o_3$. Finally, the system goes to state $s_4$, which
the agent observes with observation $o_3$.
\end{example}

We write $\append{(o,n)}$ for the observation record obtained by appending
$(o,n)$ to the observation record $r$, and
 $\rn{n}$ for the record consisting of all pairs  $(o,m)$ in $r$
 such that  $m = n$. We say that an observation record $r$ \emph{stops
   at $n$} if $\rn{m}$ is empty for all $m>n$, and $r$ \emph{stops
   at history $h$} if it stops at $|h|-1$. Unless otherwise
 specified, when we consider an
 observation record $r$ together with a history $h$, it is understood that $r$ stops at $h$.

 \vspace{1ex}
\head{Observations at time $n$} We let  $\obslist(r,n)$ be the list
of observations used by the agent  at time $n$. It consists of the
observation that the agent has when the $n$-th transition is taken,
plus those of observation changes that occur before the next
transition. It is defined by
induction on $n$:
\begin{flalign*}
\obslist(r,0)& =  \oinit \cdot o_1 \cdot \ldots \cdot o_k,& \\    
& & \mbox{if }  \rn{0}=(o_1,0)\cdot\ldots\cdot(o_k,0), \mbox{ and}\\[2pt]
 \obslist(r,n+1)&= \last(\obslist(r,n))\cdot o_1 \cdot \ldots \cdot
  o_k,& \\ 
 & & \mbox{if } \rn{n+1}=(o_1,n+1) \cdot \ldots \cdot (o_k,n+1).  
\end{flalign*}

Observe that $\obslist(r,n)$ is never empty: if no observation change
occurs at time $n$, $\obslist(r,n)$ only contains the last observation
taken by the agent. If $r$ is empty, the latter is the initial
observation $o_\init$ defined by the model.


\begin{example}
  \label{ex-2}
Consider record
$r=(o_1,0)\cdot(o_2,3)\cdot(o_3,3)$. We have that
 $\obslist(r,0)= \oinit \cdot o_1$,
  $\obslist(r,1)=\obslist(r,2)=o_1$,
   $\obslist(r,3)=o_1 \cdot o_2\cdot o_3$, and
  $\obslist(r,4)=o_3$.
\end{example}

\head{Synchronous perfect recall}
The usual definition of synchronous perfect recall states that for
an agent with observation $o$, histories $h$ and $h'$ are
indistinguishable if they have the same length and are point-wise
indistinguishable, \ie, $|h|=|h'|$ and for
each $i<|h|$, $h_i\eqstate{o} h'_i$.

We adapt this definition to   changing observations: two
histories are indistinguishable if, at
each point in time, the states are indistinguishable for all
observations used at that time. 


\begin{definition}[Dynamic synchronous perfect recall]
Given an observation record $r$, two histories $h$ and $h'$ are
equivalent, written 
  $h\eqh{r}h'$, if
\[|h|=|h'|~\textit{and}~\forall i< |h|,~\forall o\in
\obslist(r,i),~ h_i\eqstate{o} h'_i.\]
\end{definition}

We now define the natural semantics of \ctlskd.

\begin{definition}[Natural semantics]
  \label{def-nat-semantics}
 Fix a model $M$. A history formula $\phi$
is evaluated in a history $h$ and an
observation record $r$.  A path formula $\psi$ is interpreted on a
run $\pi$, a point in time $n\in\setn$
 and an observation record. The semantics is defined by induction on
 formulas: 
 \[\begin{array}{lcl}
  h,r \models p & \mbox{ if } & p\in V(\mathit{last}(h))\\[1pt]
   h,r \models \neg\varphi & \mbox{ if } & h,r\not\models\varphi\\[1pt]
   h,r \models \varphi_1\wedge\varphi_2 & \mbox{ if } &h,r\models\varphi_1~\text{and}~ h,r\models\varphi_2\\[1pt]
     h,r \models \A\psi  & \mbox{ if } & \forall\pi \mbox{
                                         s.t. }h\pref \pi,  \;\pi,|h|-1,r\models\psi\\[1pt]
  h,r \models\K\varphi  & \mbox{ if } & \forall h' \mbox{ s.t. }h'\eqh{r}h,\; h',r\models\varphi\\[1pt]
  h,r \models \D{o}\varphi & \mbox{ if } & h,\append{(o,|h|-1)}\models\varphi\\[1pt]
\pi,n,r\models\varphi & \mbox{ if } & \pi_{\leq n},r\models\varphi\\[1pt]
   \pi,n,r\models\neg\psi & \mbox{ if } & \pi,n,r\not\models\psi\\[1pt]
   \pi,n,r\models \psi_1\wedge\psi_2 & \mbox{ if }
                                       &\pi,n,r\models\psi_1 \mbox{
                                         and }\pi,n,r\models\psi_2\\[1pt]
  \pi,n,r\models\X\psi & \mbox{ if } &  \pi,(n+1),r\models\psi\\[1pt]
  \pi,n,r\models \psi_1\U\psi_2 & \mbox{ if } & 
                                                  \exists m\geq n
                                                \mbox{ s.t. }
                                                \pi,m,r\models\psi_2
                                                \mbox{ and}\\[1pt]
                        & & \forall k \mbox{ s.t. } n\leq k <m,\; \pi,k,r\models\psi_1 
\end{array}\]
\end{definition}
We say that a model $M$ with initial state $\sinit$ \emph{satisfies} a
\ctlskd formula $\phi$, written $M\models\phi$, if
$\sinit,\eobsrec\models\varphi$.

We first discuss a subtlety of our semantics, which is that an agent
can observe the same state consecutively with several
observations. 

\begin{remark}
  Consider the formula $\D{o'} \varphi$ and history $h$. By
  definition, $h,r \models \D{o'} \varphi$ iff
  $h,r \cdot (o',|h|-1) \models \varphi$.  Note that although the
  history did not change (it is still $h$), the observation record is
  extended by the observation $o'$ at time $|h|-1$, with the following
  consequence.  Suppose that $ol(r,|h|-1) = o$. After switching to
  $o'$, the agent considers possible all histories $h'$ such that i)
  $h \sim_r h'$ (they were considered possible before the change of
  observation) and ii) $last(h) \sim_{o'} last(h')$ (they are still
  considered possible after the change of observation). Informally
  this means that by changing observation from $o$ to $o'$, the
  agent's information is \emph{further refined} by $o'$,  and it is as
  though the agent at time $|h|-1$ observed the system with 
  observation $o \cap o'$. At later times, her observation is simply
  $o'$, until another change of observation occurs.   
\end{remark}

\subsection{Examples of observation change}

\begin{example}
  A logic of accumulative knowledge (and resource bounds) is introduced
  in~\cite{DBLP:journals/logcom/JamrogaT18}. It studies agents that
  can perform successive \emph{observations} to improve their knowledge of the
  situation, each observation refining their current view of the
  world. In their framework, an observation models a yes/no
  question about the current situation; if the answer is `yes', the
  agent can eliminate all possible worlds for which the answer is
  `no', and vice versa. Formally, an observation is a binary partition
  of the  possible
  states, and the agent learns in which partition is the current state.
  Such observations are particular cases of our
  models' indistinguishability relations, and the semantics of an
  agent performing an observation $o$ is
  exactly captured by the semantics of our operator
  $\D{o}$. Similarly, performing sequence of observations $o_1\ldots
  o_n$ corresponds to the successive application of operators $\D{o_1}\ldots\D{o_n}$.
  As an example, \cite{DBLP:journals/logcom/JamrogaT18} shows how to
  model a medical diagnosis in which the disease is narrowed down
  by performing a series of successive tests.
\end{example}

Our logic is incomparable with the one discussed in the previous
example, in which observations have a cost, but no temporal aspect is
considered. In this work we do not consider costs, but we study the
evolution of knowledge through time in addition to dynamic observation
change. We now illustrate with an example how both interact.




\begin{example}[Security scenario]
  \label{ex-security}
 Consider a system
 with two possible levels of security clearance, which define what
 information   users have access to, and are modelled by
observations $o_1$ and $o_2$.  In this scenario, we want to hide a
secret $p$ from the users.  A desirable property is thus expressed by the formula
$(\D{o_1} \A G \neg\K p) \wedge (\D{o_2} \A G \neg\K p)$, which means that a
user using either $o_1$ or $o_2$ will never  know that $p$
holds. Model $M$ from Figure~\ref{fig:ex} satisfies
this formula.

Now consider the formula
 $\phi=\D{o_1} \E F \D{o_2} \K p$, which means that if the user
starts with observation $o_1$, there exists a moment where she might
discover the secret by changing her observation. 
We show that $M$ satisfies $\phi$ and thus that it should forbid  users to change their security
 level. Consider the history $h=s_0s_2s_5$ in model $M$
with initial observation $o_1$.
At time 0 the user knows that the system is in state
$s_0$. After going to $s_2$, she does not know if the current state is $s_2$
or $s_1$, as they are indistinguishable using $o_1$. At time 2, at first the
user  does not know whether the system is in $s_4$ or $s_5$.
Even though $s_6$ and $s_5$ are indistinguishable with $o_1$, she does not
consider $s_6$ possible, as she remembers that the system was not in $s_3$ before.
Now, if she
changes to observation $o_2$, she sees that the system is either in state
$s_5$ or $s_6$. Refining her previous knowledge that the system is
either in state $s_4$ or $s_5$, she deduces that
the current state is $s_5$, and that $p$ holds. 
\end{example}

\begin{figure}
  \centering
  \begin{tikzpicture}[%
        every node/.style={circle,minimum size=4pt,minimum height=4pt,
          inner sep=2pt, scale=1},
        shorten >=2pt,
        node distance=0.5cm, >=latex,
        every text node part/.style={align=center}
      ]
      \node [] (0) [initial above, initial by arrow, initial text={},rectangle] {$s_0$\\$\neg p~$};
      \node [] (1) [rectangle, below=of 0] {$s_2$\\$\neg p$};
      \node [] (2) [rectangle, left=of 1,xshift=-0.8cm] {$s_1$\\$\neg p$};
      \node [] (3) [rectangle, right=of 1,xshift=0.8cm] {$s_3$\\$\neg p$};
      \node [] (4) [rectangle, below=of 1] {$s_5$\\$p$};
      \node [] (5) [rectangle, below=of 2] {$s_4$\\$\neg p$};
      \node [] (6) [rectangle, below=of 3] {$s_6$\\$\neg p$};
      \path 
      (0) edge[->]  node {} (1)
      (0) edge[->]  node {} (2)
      (0) edge[->]  node {} (3)
      (1) edge[->]  node {} (4)
      (2) edge[->]  node {} (5)
      (3) edge[->]  node {} (6)
      (4) edge[->, loop below]  node {} (4)
      (5) edge[->, loop below]  node {} (5)
      (6) edge[->, loop below]  node {} (6);
      \path (4)--node[sloped]{$\eqstate{o_1}$}(5);
      \path (1)--node[sloped]{$\eqstate{o_1}$}(2);
      \path (1)--node[sloped]{$\eqstate{o_2}$}(3);
      \path (4)--node[sloped]{$\eqstate{o_2,o_1}$}(6);
    \end{tikzpicture}
    \caption{Model $M$ in the security example}
    \label{fig:ex}
\end{figure}
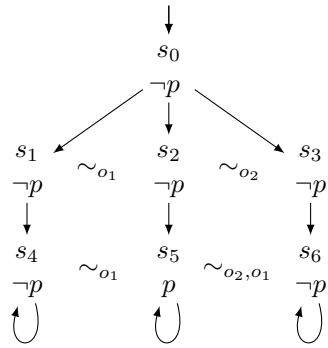


\newcommand{\sensor}{\text{sc}}

\begin{example}[Fault-Tolerant Diagnosability]
  \label{diagnos}
  Diagnosability is a property of systems in which the occurrence of a
  failure is always  eventually detected~\cite{sampath1995diagnosability}.
In the setting considered in~\cite{bittner2012symbolic},  a
\textit{diagnosability condition} is
  a pair $(c_1,c_2)$ of nonempty disjoint sets of states that the
  system should always be able to tell apart. 
  The system is monitored through a set of sensors, and they study the
  problem of finding minimal sets of sensors that ensure diagnosability.
  That is, find a minimal sensor configuration $\sensor$ such that $\D{o_{\sensor}}\A\G(\K c_1
  \vee \K c_2)$ holds, where $o_\sensor$ is the observation
  corresponding to sensor configuration $\sensor$.
  
\newcommand{\phidiag}{\Phi_{\text{diag}}}

  In \ctlskd one can express and model-check a stronger notion of
  diagnosability that we call \textit{fault-tolerant diagnosability},
  where the system must remain diagnosable even after the loss of a
  sensor.  For
  a given diagnosability condition $(c_1,c_2)$ and sensor
  configuration $\sensor$, we write $o_\sensor$ the original observation
  (with every sensor in $\sensor$), $o_i$ the observation where sensor $i$
  failed, and $p_i$ is a proposition indicating the failure of sensor
  $i$. The
  following formula expresses that sensor configuration $\sensor$
  ensures fault-tolerant diagnosability of the system:
  \[\phidiag=\D{o_\sensor}\A\G((\K c_1\vee\K c_2)\wedge (p_i\rightarrow
  \D{o_i}\A\G(\K c_1\vee\K c_2))).\]


Observe that it is possible for a system to satisfy $\phidiag$ but not
$\D{o_i}\A\G(\K c_1\vee\K c_2)$ if sensor $i$, before failing, brings some piece of
information that is crucial for diagnosis.
\end{example}

\subsection{Model-checking problem}
We are interested in the model checking-problem for \ctlskd which
consists in, given a model $M$ and a formula $\phi$, deciding whether
$M\models \phi$.

\vspace{1ex}
\head{Model-checking approach} Perfect-recall semantics refers to histories of unbounded
length, but it is well known that in many situations it is possible
to maintain a bounded amount of information that is sufficient to deal
with perfect recall. We show that it is also the case for our logic,
by generalising the classic approach to take into account observation change.
 Intuitively,  it is enough to know the current state,
the current observation and the set of states that the agent believes
the system might be in. The latter  is usually called
\emph{information set} in epistemic temporal logics and games with
imperfect information. We define an alternative
semantics based on information sets instead of
histories and records, and we prove that this semantics is equivalent
to the natural one presented in this section.
Because information sets are of bounded size, it is then easy to 
 build from this alternative semantics  a model
 checking algorithm for \ctlskd. 




\section{Alternative semantics}
\label{sec:alternative}




We define an alternative semantics for \ctlskd. It is based on information sets,
a classic notion in games with imperfect
information~\cite{von2007theory}, whose definition we now adapt to our
setting.

\begin{definition}[Information sets]
Given a model $M$, the information set $\IS(h,r)$ after a history
$h$ and an observation record $r$ is defined  as
follows:
\[\IS(h,r)=\{s\in S \mid \exists h',h'\eqh{r}h~\text{and}~\last(h')=s\}.\]  
\end{definition}

 This
information is sufficient to evaluate epistemic formulas for one
agent when we consider the S5 semantics of knowledge, \ie, when
indistinguishability relations are equivalence relations as is the
case here. We now describe how this information can be maintained along
the evaluation of a formula. To do so, we define two update functions
for information sets. One  is used
 when the agent changes of observational power, and the other
 when a transition is taken in the model.

\begin{definition}[Information set updates]
  \label{def-updates}
Fix a model
$M=(\APf,S,T,V,\{\eqstate{o}\}_{o\in\Obs},\sinit,\oinit)$.  Functions $\UT$ and
$\UD$ are defined as follows, for all
$I\subseteq S$, all $s,s'\in S$ and $o,o'\in\Obs$.
\[
\begin{array}{l}
  \UT(I,s',o)= T(I)\inter \eqc{s'}\\
  \UD(I,s,o')= I \inter \eqc[o']{s}
\end{array}
\]  
\end{definition}
These updates read as follows.  When the agent has
observational power $o$ and information set $I$, and the model takes a
transition to a state $s'$, the new information set is $\UT(I,s',o)$,
which consists of all successors of her previous information set $I$ that
are $\eqstate{o}$-indistinguishable with the new state $s'$.
When the agent is in state $s$
with information set $I$, and she changes for observational power $o'$,
her new information set is $\UD(I,s,o')$, \ie, all states that she
considered possible before  and that she still considers possible
after switching to $o'$.



We let $\lastobs(h,r)$ be the last observation taken by the agent
after history $h$, according to $r$. 
Formally, $\lastobs(h,r)=o_n$ if $\obslist(r,|h|-1)=o_1\cdot\ldots\cdot o_n$.
The following result establishes that the functions $\UD$ and $\UT$ correctly update
information sets.
 It is proved by simple application of the definitions.

\begin{proposition}
  \label{lem-UT}
  For every history $h\cdot s$, every observation record $r$ that stops
at $h$, and every observation $o$, it
  holds that
  \[
  \begin{array}{rcl}
    \IS(h\cdot s, r)&=&\UT(\IS(h,r),s,\lastobs(h,r))\mbox{,\quad and}\\
    \IS(h,r\cdot(o,|h|-1))&=&\UD(\IS(h,r),\last(h),o).
  \end{array}
  \]
\end{proposition}




Using these update functions we can now define our alternative
 semantics for \ctlskd. 

\begin{definition}[Alternative semantics]
  \label{def-alt-semantics}
   Fix a model $M$. A history formula $\phi$
is evaluated in a state $s$, an information set
$I$ and an observation $o$.  A path formula $\psi$ is interpreted on a
run $\pi$, an information set $I$ and an observation $o$. The semantic
relation $\modelsI$ is defined by induction on
 formulas: 
\[
\begin{array}{l c l}
  s,I,o\modelsI p & \mbox{if} & p\in V(s)\\[1pt]
   s,I,o\modelsI\neg\varphi & \mbox{if} & s,I,o\not\modelsI\varphi\\[1pt]
   s,I,o\modelsI \varphi_1\wedge\varphi_2 & \mbox{if}
                                &s,I,o\modelsI\varphi_1 \text{ and }s,I,o\modelsI\varphi_2\\[1pt]
  s,I,o\modelsI\A\psi & \mbox{if} & \forall\pi \mbox{ s.t. } \pi_0=s,\; \pi,I,o\modelsI\psi\\[1pt]
  s,I,o\modelsI\K\varphi & \mbox{if} & \forall s'\in I, \; s',I,o\modelsI\varphi\\[1pt]
  s,I,o\modelsI\D{o'}\varphi & \mbox{if} & s,\UD(I,s,o'),o'\modelsI\varphi\\[1pt]
  \pi,I,o\modelsI\varphi & \mbox{if} & \pi_0,I,o\modelsI\varphi\\[1pt]
   \pi,I,o\modelsI\neg\psi & \mbox{if} & \pi,I,o\not\modelsI\psi\\[1pt]
   \pi,I,o\modelsI\psi_1\wedge\psi_2 & \mbox{if}
                                &\pi,I,o\modelsI\psi_1 \text{ and } \pi,I,o\modelsI\psi_2\\[1pt]
  \pi,I,o\modelsI\X\psi & \mbox{if} & \pi_{\geq 1},\UT(I,\pi_1,o),o\modelsI\psi\\[1pt]
  \pi,I,o\modelsI\psi_1\U\psi_2 & \mbox{if} & \exists n\geq 0 \mbox{
                                             such that }
 \pi_{\geq n},\; \UT^n(I,\pi,o),o\modelsI\psi_2 \mbox{ and}\\[1pt]
                 & & \forall m \mbox{ such that }0\leq m <n, \; \pi_{\geq m},\UT^m(I,\pi,o),o\modelsI\psi_1,
\end{array}
\]
where $\UT^n(I,\pi,o)$ is the iteration of the temporal update,
defined
inductively as follows:
\begin{itemize}
\item $\UT^0(I,\pi,o)=I$, and
\item $\UT^{n+1}(I,\pi,o)=\UT(\UT^n(I,\pi,o),\pi_{n+1},o)$.
\end{itemize}

\end{definition}

Using Proposition~\ref{lem-UT}, one can prove that the natural semantics $\models$ and the information
semantics $\modelsI$ are equivalent (the proof is in Appendix~\ref{sec-app}).

\newcounter{theo-equiv}
\setcounter{theo-equiv}{\value{theorem}}

\begin{theorem}
  \label{theo-equivalence}
  For every history formula $\phi$, model $M$, history $h$ and
  observation record $r$ that stops at $h$,
  \[h,r\models\phi \quad\mbox{iff}\quad
    \last(h),\IS(h,r),o(h,r)\modelsI \phi.\]
\end{theorem}





\section{Model checking \ctlskd}
\label{sec:mc}


In this section we devise a
model-checking procedure based on the equivalence between the natural and alternative
semantics (Theorem~\ref{theo-equivalence}), and we  prove the following result.

\begin{theorem}
  \label{theo-compl}
Model checking  \ctlskd is in \EXPTIME. 
\end{theorem}

\head{Augmented model}
Given a model $M$, we define an augmented model $\hatM$ in which the
states are tuples $(s,I,o)$ consisting of a state $s$ of the original
model, an information set $I$ and an observation $o$.
According to Theorem~\ref{theo-equivalence}, history formulas can be
viewed on this model as state formulas, and a model checking
procedure can be devised by merely  following the definition of the alternative semantics.

Let
$M=(\APf,S,T,V,\{\eqstate{o}\}_{o\in\Obs},\sinit,\oinit)$. 
We define the Kripke structure $\hatM=(S',T',V',\sinit')$, where:
\begin{itemize}
\item $S'=S\times 2^{S}\times \mathcal{O}$,
\item $(s,I,o)~T'~(s',I',o)\mbox{ if } s~T~s'$ and $I'=\UT(I,s',o)$,
\item $V'(s,I,o)=V(s)$, and 
\item $\sinit' = (\sinit,\eqc[\oinit]{\sinit},\oinit)$.
\end{itemize}
We call $\hatM$ the \emph{augmented model}, and we write $\hatM_o$
the Kripke structure obtained by restricting $\hatM$
to states of the form $(s,I,o')$ where $o'=o$. Note that the different $\hatM_o$ are disjoint with regards to $T'$.

\vspace{1ex}
\head{Model-checking procedure}
We now define the function \textsc{Check}\ctlskd\
which evaluates a history formula in a  state of the augmented model $\hatM$:  

Function
\textsc{Check}\ctlskd($\hatM,(s_c,I_c,o_c),\Phi$) returns
\emph{true} if $M,s_c,I_c,o_c\modelsI\varphi$ and \emph{false}
otherwise, and is defined
as follows: if $\Phi$ is a \ctls formula, we evaluate it using a
classic model-checking procedure for \ctls. Otherwise, $\Phi$ contains a subformula
of the form $\phi=\K\varphi_1$ or $\phi=\D{o'}\varphi_1$ where
$\varphi_1\in \ctls$. We evaluate $\phi_1$ in every state of 
$\hatM$, and 
mark those that satisfy $\phi_1$ with a fresh atomic proposition
$p_{\varphi_1}$.  Then, if $\phi=\K\varphi_1$, we mark with a fresh atomic
proposition $p_\phi$ every  state $(s,I,o)$ of $\hatM$ such that
for every $s'\in I$, $(s',I,o)$ is marked with $p_{\varphi_1}$.  Else, 
$\phi=\D{o'}\varphi_1$ and we mark with a fresh proposition $p_\phi$
every state $(s,I,o)$  such that
$(s,\UD(I,s,o'),o')$ is marked with $p_{\varphi_1}$. Finally, we
recursively call function \textsc{Check}\ctlskd on the marked model and formula
$\Phi'$ obtained by replacing $\phi$ with $p_\phi$ in $\Phi$. 

To model check a formula $\phi$ in a model $M$, we first build
$\hatM$,  and call function
\textsc{Check}\ctlskd($\hatM,(s_\init,\eqc[o_\init]{s_\init},o_\init),\phi$).

\begin{example}
  \label{ex-algo}
Let $M$ be the model depicted in \textbf{Fig.}~\ref{fig:m}, where 
$o_1$ is the blind observation  ($\eqstate{o_1}=S\times S$),
$o_2$ is the perfect observation
($s\eqstate{o_2}s'\iff s=s'$),  $s_1$  is the initial state and the
agent is initially blind (the initial observation is $o_1$). Note that we did not
represent indistinguishability relations $o_1$ and $o_2$. The augmented model $\hatM$ is depicted in
 \textbf{Fig.}~\ref{fig:hatm}, where we only drew
relevant states, \ie those that are
reachable from the initial state via transitions and observation changes. 

Consider  formula $\varphi=\D{o_2}(\K q\vee\D{o_1}\K\A\X q)$ which  means that if the agent changes to the perfect
observation, then either the agent knows that $q$ holds, or even after
switching back to the blind observation she knows that in every possible next
step, $q$ holds.  After running the algorithm, we get the following
valuation for $\hatM$:

\vspace{2pt}
\noindent$\begin{array}{l c l}
V'(s_1,\{s_1,s_2\},o_1) &=& \{q,p_\varphi\}\\
V'(s_2,\{s_1,s_2\},o_1) &=& \{p_\varphi \}
\end{array}\newline
\begin{array}{l c l}
V'(s_1,\{s_1\},o_1) &=& \{q,p_{(\K q)},p_{(\K q\vee\D{o_1}\K\A\X q)},p_\varphi\}\\
V'(s_2,\{s_2\},o_1) &=& \{p_{(\D{o_1}\K\A\X q)},p_{(\K q\vee\D{o_1}\K\A\X q)},p_\varphi\}\\
V'(s_1,\{s_1\},o_2) &=& \{q,p_{(\K q)},p_{(\K q\vee\D{o_1}\K\A\X q)},p_\varphi\}\\
V'(s_2,\{s_2\},o_2) &=& \{p_{(\D{o_1}\K\A\X q)},p_{(\K q\vee\D{o_1}\K\A\X q)},p_\varphi\}\\
\end{array}$

\vspace{2pt}

Let us explain why $p_\phi$ is eventually marked in the initial state
$(s_1,\{s_1,s_2\},o_1)$ of $\hatM$. First look at state $(s_1,\{s_1\},o_2)$. Since $q$ is in $V(s_1)$, 
 $q$ is also in $V'(s_1,\{s_1\},o_2)$. Next,  because
for all $s'$ in $\{s_1\}$, $q$ is in $V'(s',\{s_1\},o_2)$, the fresh atom
$p_{(\K q)}$ is added in $V'(s_1,\{s_1\},o_2)$ and thus the fresh atom 
 $p_{(\K q\vee\D{o_1}\K\A\X q)}$ is later also added in $V'(s_1,\{s_1\},o_2)$.
 Finally, when $\D{o_2}(\K q\vee \D{o_1}\K\A\X q)$ is evaluated on
 $(s_1,\{s_1,s_2\},o_1)$,
 since $\UD(\{s_1,s_2\},s_1,o_2)=\{s_1\}$ and
  $p_{(\K q\vee\D{o_1}\K\A\X q)}$ is in $V'(s_1,\{s_1\},o_2)$, the
  fresh atom
$p_\varphi$ is eventually added in $V'(s_1,\{s_1,s_2\},o_1)$.

The algorithm thus returns, as expected, that $M\models\varphi$.

\begin{minipage}{.25\linewidth}
  \centering
  \begin{tikzpicture}[%
    scale=0.9,
        every node/.style={circle,minimum size=4pt,minimum height=4pt,
          inner sep=2pt, scale=0.9},
        shorten >=2pt,
        node distance=1.3cm, >=latex,
        every text node part/.style={align=center}
      ]
      \node [initial below, initial by arrow, initial text={}] (0) [circle] {$s_1$\\$q$};
      \node [] (1) [circle, right=of 0] {$s_2$\\$\neg q$}; 
      \path [draw] (0) edge[->, bend right]  node {} (1)
      (0) edge[->, loop above]  node {} (0)
      (1) edge[->, bend right]  node {} (0);
    \end{tikzpicture}
    \captionof{figure}{Model $M$}
    \label{fig:m}
\end{minipage}~
\begin{minipage}{.7\linewidth}%
  \centering
  \begin{tikzpicture}[%
    scale=0.8,
        every node/.style={circle,minimum size=4pt,minimum height=4pt,
          inner sep=2pt, scale=0.8},
        shorten >=2pt,
        node distance=1cm, >=latex,
        every text node part/.style={align=center}
      ]
      \node [] (0) [initial above, initial by arrow, initial text={}, rectangle] {$s_1,\{s_1,s_2\},o_1$\\$q$};
      \node [] (1) [rectangle, right=of 0] {$s_2,\{s_1,s_2\},o_1$\\$\neg q$};
      \node [] (2) [rectangle, below=of 0] {$s_1,\{s_1\},o_1$\\$q$};
      \node [] (3) [rectangle, below=of 1] {$s_2,\{s_2\},o_1$\\$\neg q$};
      \node [] (4) [rectangle, below=of 2] {$s_1,\{s_1\},o_2$\\$q$};
      \node [] (5) [rectangle, below=of 3] {$s_2,\{s_2\},o_2$\\$\neg q$};
      \path [draw] (0) edge[->, bend right]  node {} (1)
      (0) edge[->, loop left]  node {} (0)
      (1) edge[->, bend right]  node {} (0)
      (4) edge[->, bend right]  node {} (5)
      (4) edge[->, loop left]  node {} (4)
      (5) edge[->, bend right]  node {} (4)
      (2) edge[->]  node {} (0)
      (2) edge[->, bend right]  node {} (1)
      (3) edge[->, bend left]  node {} (2);
    \end{tikzpicture}%
  \captionsetup{width=.7\linewidth}
  \captionof{figure}{The augmented model $\hatM$}
    \label{fig:hatm}
\end{minipage}
\end{example}

\head{Algorithm correctness}
The correctness of the algorithm follows from the following properties:
\begin{itemize}
\item For each formula $\K\varphi_1$ chosen by the algorithm, \\$p_{\phi}\in V'(s,I,o)\iff M,s,I,o\modelsI\K\varphi_1$
\item For each formula $\D{o'}\varphi_1$ chosen by the algorithm, \\$p_{\phi}\in V'(s,I,o)\iff M,s,I,o\modelsI\D{o'}\varphi_1$
\end{itemize}



\vspace{1ex}
\head{Complexity analysis} 
In the following, we let $|M|$ be the number of states in model $M$.
Model checking a \ctls formula $\phi$ on a model $M$ with state-set
$S$ can be done in time
$2^{O(|\phi|)}O(|S|)$~\cite{emerson1987modalities,kupferman2000automata}.
Our procedure, for a \ctlskd formula $\phi$ and a model $M$, calls the
\ctls model-checking procedure for at most $|\phi|$ formulas of size
at most $|\phi|$, on each state of the augmented model $\hatM$. The latter is of size $2^{O(|M|)}\times |\Obs|$, but each
 call to the \ctls model-checking procedure is performed on a disjoint component $\hatM_o$
of size $2^{O(|M|)}$. Our overall procedure thus runs in time
$|\Obs|\times 2^{O(|\phi|+|M|)}$.


\section{Multi-agent setting}
\label{sec:multi}
We now extend \ctlskd to the multi-agent setting.  We fix $\Ag=\{a_1,\ldots,a_m\}$
a finite set of agents and define the logic \ctlskdm. This logic contains, for each agent $a$ and
observation $o$, an operator $\Da{o}$
which reads as ``agent $a$ changes for observation $o$''. We consider
that these observation changes are public
in the sense that all agents
are aware of them.  The reason is that if agent $a$ changes
observation without agent $b$ knowing it, agent $b$ may entertain false
beliefs about what agent $a$ knows. This would not be consistent with the S5
semantics of knowledge that we consider in this work, where false beliefs are ruled out by the Truth
axiom $\K\phi \impl \phi$. 

\subsection{Syntax and natural semantics}

We first extend the syntax, with knowledge operators $\K_a$ and
observation change operators $\Da{o}$ for each agent.

\begin{definition}[Syntax]
The sets of history formulas $\phi$ and path formulas $\psi$ are
defined by the following grammar:
\[
\begin{array}{ccl}
  \phi & ::= & p ~|~ \neg \varphi ~|~ \varphi\wedge\varphi ~|~ \A\psi
               ~|~ \K_a\varphi ~|~ \Da{o}\varphi\\
  \psi & ::= & \varphi ~|~ \neg\psi ~|~ \psi\wedge\psi ~|~ \X\psi ~|~ \psi\U\psi,
\end{array}
\]
where $p\in\AP$, $a\in\Ag$ and $o\in\Obs$. 
\end{definition}

Formulas of \ctlskdm  are all history formulas.

The models of \ctlskdm are as for the one-agent case, except that we
assign one initial observation to each agent. In the following we shall write $\ovec$ for a tuple $\{o_a\}_{a\in\Ag}$, 
$\ovec_a$ for $o_a$, and $\ovec[a\leftarrow o]$ for the tuple $\ovec$
where $\ovec_a$ is replaced with $o$. Finally, for $i\in \{1,\ldots,m\}$, $\ovec_i$ refers to $\ovec_{a_i}$.

\begin{definition}[Multiagent models]
  \label{def-multi-model}
A \emph{multiagent Kripke
structure with observations} is given by a structure
$M=(\APf,S,T,V,\{\eqstate{o}\}_{o\in\Obs},\sinit,\ovecinit)$, where
\begin{itemize}
  \item $\APf\subset \AP$  is a finite subset of atomic propositions,
\item $S$ is a set of states,
\item $T\subseteq S\times S$ is a left-total transition relation between states,
\item $V:S\rightarrow 2^{\APf}$ is a valuation function, 
\item $\eqstate{o}\subseteq S\times S$ is an equivalence relation, for each
  $o\in\Obs$, 
 \item $\sinit\subseteq S$ is an initial state, and
\item $\ovecinit$ is the initial observation for each agent.
\end{itemize}
\end{definition}

We now adapt some definitions to the multi-agent setting.

\vspace{1ex}
\head{Records tuples} We now need one
observation record for each agent. We shall write $\rvec$ for a tuple
$\{r_a\}_{a\in\Ag}$. Given a tuple $\rvec=\{r_a\}_{a\in\Ag}$ and $a\in\Ag$ we
write $\rvec_a$ for $r_a$, and for an observation $o$ and time $n$ we
let $\rvec\cdot (o,n)_a$ be the record tuple $\rvec$ where $\rvec_a$ is
replaced with $\rvec_a\cdot (o,n)$.
Finally, for $i\in \{1,\ldots,m\}$, $\rvec_i$ refers to $\rvec_{a_i}$.

 \vspace{1ex}
\head{Observations at time $n$} We let  $\obslista(\rvec,n)$ be the list
of observations used by agent $a$  at time $n$:
\begin{flalign*}
  \obslista(\rvec,0)& =  \ovecinit_a \cdot o_1 \cdot \ldots \cdot
  o_k,& \\[1pt]
 & & \mbox{if }  \rn[\rvec_a]{0}=(o_1,0)\cdot\ldots\cdot(o_k,0), \mbox{ and}\\[2pt]
 \obslista(\rvec,n+1)&= \last(\obslista(\rvec,n))\cdot o_1 \cdot \ldots \cdot o_k,&\\[1pt]
& & \mbox{if } \rn[\rvec_a]{n+1}=(o_1,n+1) \cdot \ldots \cdot (o_k,n+1).
\end{flalign*}


\begin{definition}[Dynamic synchronous perfect recall]
Given a record tuple $\rvec$, two histories $h$ and $h'$ are
equivalent for agent $a$, written 
  $h\eqha{\rvec}h'$, if
\[|h|=|h'|~\textit{and}~\forall i< |h|,~\forall o\in
\obslista(\rvec,i),~ h_i\eqstate{o} h'_i.\]
\end{definition}

\begin{definition}[Natural semantics]
  \label{def-nat-semantics-multi}
Let $M$ be a model, $h$ a history in $M$ and $\rvec$ a record tuple.
We only define the semantics for the following inductive cases, the
remaining ones are straightforwardly adapted from the one-agent case
(Definition~\ref{def-nat-semantics}).
\[
\begin{array}{lcl}
h,\rvec \models\K_a\varphi & \mbox{if} & \forall h' \mbox{ s.t. }h'\eqha{\rvec}h,\; h',\rvec\models\varphi\\[2pt]
h,\rvec \models \D{o}_a\varphi& \mbox{if} & h,\rvec\cdot(o,|h|-1)_a\models\varphi
\end{array}
\]
\end{definition}
A model $M$ with initial state $\sinit$  \emph{satisfies} a \ctlskdm formula $\phi$, written $M\models\phi$, if 
$\sinit,\eobsrecvec\models\varphi$, where $\eobsrecvec$ is the tuple
where each agent has empty observation record.

\subsection{Alternative semantics}
As in the one-agent case, we define an alternative  semantics
that we prove equivalent to the natural one and upon which we build our
model-checking algorithm.
The main difference here is that we need richer structures than
information sets to represent an epistemic situation of
a  system with multiple agents. For instance, to evaluate formula
$\K_a\K_b\K_cp$, we need to know what agent $a$ knows about agent
$b$'s knowledge of agent $c$'s knowledge of the system's state.
To do so we use the  \ktrees introduced
in~\cite{DBLP:journals/iandc/Meyden98,DBLP:conf/fsttcs/MeydenS99}
 in the setting of static observations, and which contain enough
information to evaluate formulas of knowledge depth $k$.


\vspace{1ex}
\head{\ktrees}
Fix a model $M=(\APf,S,T,V,\{\eqstate{o}\}_{o\in\Obs},\sinit,\ovecinit)$. Intuitively, a
\ktree over $M$ is a structure of the form $\langle
s, \forest_1,\ldots,\forest_m\rangle$, where
$s\in S$ is the current state of the system, and for each
$i\in \{1,\ldots,m\}$, $\forest_i$ is a set of $(k-1)$-trees that represents the state of knowledge
(of depth $k-1$)
of  agent $a_i$.  Formally, for every history $h$ and record tuple $\rvec$ we define by
induction on $k$ the $k$-tree $\KT{k}(h,\rvec)$ as follows:
\begin{align*}
\KT{0}(h,\rvec)&\egdef \langle \last(h),\emptyset,\ldots,\emptyset
  \rangle \\
\KT{k+1}(h,\rvec) &\egdef \langle
  \last(h),\forest_1,\ldots,\forest_m\rangle,
\end{align*}
where for each $i$, $\forest_i\egdef\{\KT{k}(h',\rvec) \mid h'
\eqha[a_i]{\rvec} h\}$.

For a $k$-tree $\KT{k}=\langle \state, \forest_1, \ldots, \forest_m
\rangle$, we call $\state$ the \emph{root} of $\KT{k}$, and 
write it $\racine(\KT{k})$. We also write $\KT{k}(a)$ for
$\forest_{i}$, where $a=a_i$, and we let $\setktrees$ be the set of
$k$-trees for $M$.

Observe that for one agent (\ie $m=1$), a $1$-tree is an information
set together with the current state.





\vspace{1ex}
\head{Updating \ktrees}
We now generalise our update functions $\UD$ and $\UT$
(Definition~\ref{def-updates}) to update \ktrees. We first define, by
induction on $k$, the function
 $\UTK{k}$ that  updates \ktrees when a transition
is taken.
\begin{align*}
  \UTK{0}(\langle s, \emptyset, \ldots,\emptyset\rangle,s',\ovec) & \egdef \langle s', \emptyset,
                          \ldots,\emptyset \rangle \\
  \UTK{k+1}(\langle s, \forest_1, \ldots,\forest_m\rangle,s',\ovec) & \egdef \langle s', \forest'_1,
                                  \ldots, \forest'_m\rangle,
\end{align*}
where for each $i$, \[\forest'_i\egdef \{\UTK{k}(\KT{k},s'',\ovec)\mid
\KT{k}\in\forest_i\mbox{, }s'' \eqstate{\ovec_i} s' \mbox{ and
}\racine(\KT{k})T s''\}.\]

$\UTK{k}$ takes the current \ktree $\langle s, \forest_1,
\ldots,\forest_m\rangle$, the new state $s'$ and the current
observation  $\ovec$ for each agent, and returns the new $k$-tree after
the transition.

We now define the second update function $\UDK{k}$, which is used when
an agent $a_i$ changes observation for some $o'$.
\begin{align*}
  \UDK{0}(\langle s, \emptyset, \ldots, \emptyset\rangle,o,a_i) & \egdef \langle s, \emptyset, \ldots, \emptyset\rangle \\
  \UDK{k+1}(\langle s, \forest_1, \ldots,\forest_m\rangle,o,a_i) & \egdef \langle s, \forest'_1,
                                  \ldots,\forest'_m\rangle,
\end{align*}
where for each $j\neq i$, \[\forest'_j=\{\UDK{k}(\KT{k},o',a_i)\mid
  \KT{k}\in\forest_j\}, \mbox{ and}\]
\[\forest'_i\egdef \{\UDK{k}(\KT{k},o',a_i)\mid
\KT{k}\in\forest_i \mbox{ and } \racine(\KT{k})\eqstate{o'} s\}.\]

The intuition is that when agent~$a_i$ changes observation for $o'$, 
 in every place of the \ktree that refers to agent~$a_i$'s knowledge,
 we remove possible states (and corresponding subtrees) that are no longer equivalent to the current
possible state for $a_i$'s new observation $o'$.

We let $\lastovec(h,\rvec)$ be the tuple of last observations taken by each agent
after history $h$, according to $\rvec$.
For each $a\in\Ag$, $\lastovec(h,\rvec)_a=o_n$ if $\obslista(\rvec,|h|-1)=o_1\cdot\ldots\cdot o_n$.
The following proposition establishes that functions $\UTK{k}$ and $\UDK{k}$ correctly update \ktrees.

\newcounter{lem-UTk}
\setcounter{lem-UTk}{\value{theorem}}
\begin{proposition}
  \label{lem-UTk}
For every history $h\cdot s$, record tuple $\rvec$ that stops at $h$, observation tuple $\ovec$  and 
integer $k$, it holds that
\[
  \begin{array}{rcl}
  \KT{k}(h\cdot s,\rvec)&=&\UTK{k}(\KT{k}(h,\rvec),s,\ovec(h,\rvec)),\mbox{ and}\\
    \KT{k}(h,\append[\rvec]{(o,|h|-1)_a})&=&\UDK{k}(\KT{k}(h,\rvec),o,a).
  \end{array}
\]
\end{proposition}

The proof can be found in Appendix~\ref{sec-app-b}.

We now define the alternative semantics for \ctlskdm.

\begin{definition}[Alternative semantics]
  \label{def-alt-semantics-multi}
The semantics of a history formula $\phi$ of knowledge depth $k$ is
defined inductively on a \ktree $\KT{k}$ and a tuple of current
observations $\ovec$ (note that the
current state is the root of the \ktree). 
\[
\begin{array}{l c l}
  \KT{k},\ovec\modelsI p & \mbox{if} & p\in V(\racine(\KT{k}))\\[2pt]
   \KT{k},\ovec\modelsI\neg\varphi & \mbox{if} & \KT{k},\ovec\not\modelsI\varphi\\[2pt]
   \KT{k},\ovec\modelsI \varphi_1\wedge\varphi_2 & \mbox{if}
                                &\KT{k},\ovec\modelsI\varphi_1 \text{ and }\KT{k},\ovec\modelsI\varphi_2\\[2pt]
   \KT{k},\ovec\modelsI\A\psi & \mbox{if} & \forall\pi \mbox{ s.t. } \pi_0=\racine(\KT{k}),\; \pi,\KT{k},\ovec\modelsI\psi\\[2pt]
  \KT{k},\ovec\modelsI\K_a\varphi & \mbox{if} & \forall \KT{k-1}\in \KT{k}(a), \; \KT{k-1},\ovec\modelsI\varphi\\[2pt]
  \KT{k},\ovec\modelsI\D{o'}_a\varphi & \mbox{if} &
                                                    \UDK{k}(\KT{k},o',a),\ovec[a\leftarrow
                                                    o']\modelsI\varphi\\
  & & \quad\quad\quad \mbox{where }\ovec[a\leftarrow o'] \mbox{ is }\ovec
      \mbox{ where }\ovec_a \mbox{ is replaced with }o'\\[2pt]
   \pi,\KT{k},\ovec\modelsI\varphi & \mbox{if} & \KT{k},\ovec\modelsI\varphi\\[2pt]
   \pi,\KT{k},\ovec\modelsI\neg\psi & \mbox{if} & \pi,\KT{k},\ovec\not\modelsI\psi\\[2pt]
   \pi,\KT{k},\ovec\modelsI\psi_1\wedge\psi_2 & \mbox{if}
                                &\pi,\KT{k},\ovec\modelsI\psi_1 \text{ and } \pi,\KT{k},\ovec\modelsI\psi_2\\[2pt]
   \pi,\KT{k},\ovec\modelsI\X\psi & \mbox{if} & \pi_{\geq 1},\UTK{k}(\KT{k},\pi_1,\ovec),\ovec\modelsI\psi\\[2pt]
   \pi,\KT{k},\ovec\modelsI\psi_1\U\psi_2 & \mbox{if} & \exists n\geq 0 \mbox{
                                              such that }
 \pi_{\geq n},{\UTK{k}}^n(\KT{k},\pi,\ovec),\ovec\modelsI\psi_2 \mbox{ and}\\[1pt]
                  & & \forall m \mbox{ such that }0\leq m <n,\; \pi_{\geq m},{\UTK{k}}^m(\KT{k},\pi,\ovec),\ovec\modelsI\psi_1,
\end{array}
\]
where ${\UTK{k}}^n$ is the iteration of $\UTK{k}$,
defined similarly to the mono-agent case.
\end{definition}

The following theorem can be proved similarly to
Theorem~\ref{theo-equivalence}, using Proposition~\ref{lem-UTk}
instead of Proposition~\ref{lem-UT}.

\begin{theorem}
  \label{theo-equivalence-multi}
  For every history formula $\phi$ of knowledge depth $k$, each model $M$, history $h$ and
tuple of records $\rvec$,
  \[h,\rvec\models\phi \quad\mbox{iff}\quad
    \KT{k}(h,\rvec),\ovec(h,r)\modelsI \phi.\]
\end{theorem}




\section{Model checking \ctlskdm}
\label{sec:mc-m}


Like in the mono-agent case, it is rather easy to devise from this alternative
semantics a model-checking algorithm for \ctlskdm, the main difference
 being that the states of
the augmented model are now $k$-trees. 
In this section we adapt the model-checking procedure for \ctlskd to
the multi-agent setting, once again relying
 on the equivalence between the natural and alternative
semantics (Theorem~\ref{theo-equivalence-multi}), and we  prove the following result.

\begin{theorem}
  \label{theo-complex-multi}
  The model-checking problem for \ctlskdm is in $k$-\textsc{EXPTIME} for formulas
  of knowledge depth at most $k$.
\end{theorem}

\head{Augmented model}
Given a model $M$, we define an augmented model $\hatM$ in which the
states are pairs $(\KT{k},\ovec)$ consisting of a $k$-tree $\KT{k}$ and an observation for each agent, $\ovec$.

Let
$M=(\APf,S,T,V,\{\eqstate{o}\}_{o\in\Obs},\sinit,\ovecinit)$. 
We define the Kripke structure $\hatM=(S',T',V',\sinit')$, where:
\begin{itemize}
\item $S'=\setktrees\times \mathcal{O}^\Ag$,
\item
  $(\KT{k},\ovec)~T'~({\KT{k}}',\ovec)\mbox{ if }
  s~T~s'$ and
   ${\KT{k}}'=\UTK{k}(\KT{k},s',\ovec)$, where $s=\racine(\KT{k})$ and
   $s'=\racine({\KT{k}}')$, 
\item $V'(\KT{k},\ovec)=V(\racine(\KT{k}))$, and 
  \item $\sinit' = (\KT{k}(\sinit,\eobsrecvec),\ovecinit)$. 
\end{itemize}
We call $\hatM$ the \emph{augmented model}, and we write $\hatMovec$
the Kripke structure obtained by restricting $\hatM$ to states of the
form $(\KT{k},\ovec')$ where $\ovec'=\ovec$. Note that the different
$\hatMovec$ are disjoint with regards to $T'$.

\vspace{1ex}
\head{Model-checking procedure}
We now define the function \textsc{Check}\ctlskdm\
which evaluates a history formula in a  state of the augmented model $\hatM$:  

Function
\textsc{Check}\ctlskdm($\hatM,(\KT{k}_c,\ovec_c),\Phi$) returns
\emph{true} if $M,\KT{k}_c,\ovec_c\modelsI\varphi$ and \emph{false}
otherwise, and is defined
as follows: if $\Phi$ is a \ctls formula, we evaluate it using a
classic model-checking procedure for \ctls. Otherwise, $\Phi$ contains a subformula
of the form $\phi=\Ka\varphi'$ or $\phi=\Da{o'}\varphi'$ where
$\varphi'\in \ctls$. We evaluate $\phi'$ in every state of 
$\hatM$, and 
mark those that satisfy $\varphi'$ with a fresh atomic proposition
$p_{\varphi'}$.  Then, if $\phi=\Ka\varphi'$, we mark with a fresh atomic
proposition $p_\phi$ every  state $(\KT{k},\ovec)$ of $\hatM$ such that
for every $\KT{k-1}\in \KT{k}(a)$, $(\KT{k-1},\ovec)$ is marked with $p_{\varphi'}$.  Else, 
$\phi=\Da{o'}\varphi'$ and we mark with a fresh proposition $p_{\phi}$
every state $(\KT{k},\ovec)$  such that
$(\UDK{k}(\KT{k},o',a),\ovec[a \leftarrow o'])$ is marked with $p_{\varphi'}$. Finally, we
recursively call function \textsc{Check}\ctlskdm on the marked model and formula
$\Phi'$ obtained by replacing $\phi$ with $p_\phi$ in $\Phi$. 

To model check a formula $\phi$ in a model $M$, we build $\hatM$ and call 
\textsc{Check}\ctlskdm($\hatM,(\KT{k}(\sinit,\eobsrecvec),\ovecinit),\phi$).

\vspace{1ex}
\head{Algorithm correctness}
The correctness of the algorithm follows from the following properties:
\begin{itemize}
\item For each formula $\Ka\varphi$ chosen by the algorithm,
  \\$p_{\phi}\in V'(\KT{k},\ovec)\iff
  M,\KT{k},\ovec\modelsI\Ka\varphi$
\item For each formula
  $\Da{o'}\varphi$ chosen by the algorithm,
  \\$p_{\phi}\in V'(\KT{k},\ovec)\iff M,\KT{k},\ovec\modelsI\Da{o'}\varphi$
\end{itemize}



\vspace{1ex}
\head{Complexity analysis} 
The number of different
\ktrees for $m$ agents and a model with $l$ states is no
greater than $C_k=\tower{m\times l}{k}/m$, where $\tower{a}{b}$ is defined
as $\tower{a}{0}=a$ and $\tower{a}{b+1}=a
2^{\scriptsize\tower{a}{b}}$~\cite{DBLP:conf/fsttcs/MeydenS99}.
The size of the augmented model $\hatM$  is thus bounded by $\tower{m\times l}{k}/m\times |\Obs|^{|\Ag|}$, and it can be
computed in time $\tower{O(m\times l)}{k}\times
|\Obs|^{|\Ag|}$. 

Model checking a \ctls formula $\phi$ on a model $M$ with state-set
$S$ can be done in time
$2^{O(|\phi|)}\times O(|S|)$~\cite{emerson1987modalities,kupferman2000automata}.
For a \ctlskdm formula $\phi$ of knowledge depth at most $k$ and a
model $M$ with $l$ states, our
procedure calls the
\ctls model-checking procedure for at most $|\phi|$ formulas of size
at most $|\phi|$, on each state of the augmented model $\hatM$ which has size
$\tower{m\times l}{k}/m\times |\Obs|^{m}$.
Each recursive call (for each subformula and state of $\hatM$) is 
performed  on a disjoint component $\hatMovec$ of size at most $\tower{m\times l}{k}/m$,
and thus takes time $2^{O(|\phi|)}\times O(\tower{m\times
  l}{k}/m)$, and there are at most $|\phi|\times
\tower{m\times l}{k}/m\times |\Obs|^{m}$ of them.
Our overall procedure thus runs in time $
|\Obs|^{m}\times2^{O(|\phi|)}\times\tower{O(m\times l)}{k}$, \ie
 $|\Obs|^{|\Ag|}\times 2^{O(|\phi|)}\times\tower{O(|\Ag|\times |M|)}{k}$.

\section{Expressivity}
\label{sec:expressivity}
In this section we prove that the observation-change operator adds
expressive power to epistemic temporal logics. Formally, we compare
the expressive power of \ctlskdm with that of
\ctlskm~\cite{DBLP:journals/siamcomp/HalpernMV04,DBLP:conf/fossacs/BozzelliMP15},
which is the syntactic fragment of \ctlskdm obtained by removing the
observation-change operator. 
Our semantics for \ctlskdm generalises that of \ctlskm, with which it
coincides on \ctlskm formulas. Note that our multi-agent models
 (Definition~\ref{def-multi-model}) are more general than usual models for \ctlskm, as they
 may contain observation
relations that are not initially assigned to any agent, but such relations
are mute in the evaluation of \ctlskm formulas.


For two logics $\lang$ and $\lang'$ over the same class of models, we say that $\lang'$ is \emph{at least as expressive as}
$\lang$, written $\lang\lexpr\lang'$, if for every formula
$\phi\in\lang$ there exists a formula $\phi'\in\lang'$ such that $\phi\equiv\phi'$. $\lang'$
is \emph{strictly more expressive than} $\lang$,  written $\lang\slexpr\lang'$,  if $\lang\lexpr\lang'$
and $\lang'\not\lexpr\lang$. Finally, $\lang$ and $\lang'$ are
\emph{equiexpressive}, written $\lang\equiexpr\lang'$, if
$\lang\lexpr\lang'$ and $\lang'\lexpr\lang$.

\begin{proposition}
  \label{prop-more-expressive}
For all $m\geq 1$,  $\ctlskm \lexpr \ctlskdm$.
\end{proposition}

\begin{proof}
This simply follows from the fact that \ctlskdm extends
 \ctlskm. 
\end{proof}

We first point out that when there is only one observation, \ie,
$|\Obs|=1$, the observation-change operator has no effect, and thus
\ctlskdm is no more expressive than \ctlskm.

\begin{proposition}
  \label{prop-equiexpr}
  For $|\Obs|=1$, $\ctlskm\equiexpr\ctlskdm$.
\end{proposition}

\begin{proof}
  We show that  for $|\Obs|=1$, $\ctlskdm\lexpr\ctlskm$, which
  together with Proposition~\ref{prop-more-expressive} provides the
  result.
  Observe that when $|\Obs|=1$, observation change has no effect, and
  in fact observation records can be omitted in
  the natural semantics. For every
  \ctlskdm formula  $\phi$, define the \ctlskm formula
  $\phi'$ by removing all observation-change operators $\Da{o}$ from
  $\phi$. It is easy to see that $\phi\equiv\phi'$.
\end{proof}

We now show that as soon as we have at least two observations, 
the observation-change operator adds expressivity. 
We first consider the mono-agent case.

\begin{figure}
  \centering
  \begin{tikzpicture}[%
        every node/.style={circle,minimum size=4pt,minimum height=4pt,
          inner sep=2pt, scale=1},
        shorten >=2pt,
        node distance=0.5cm, >=latex,
        every text node part/.style={align=center}
      ]
      \node [] (0) [initial above, initial by arrow, initial text={},rectangle] {$s_0$\\$\neg p~$};
      \node [] (1) [rectangle, below=of 0] {$s_2$\\$\neg p$};
      \node [] (2) [rectangle, left=of 1,xshift=-0.8cm] {$s_1$\\$\neg p$};
      \node [] (3) [rectangle, right=of 1,xshift=0.8cm] {$s_3$\\$\neg p$};
      \node [] (4) [rectangle, below=of 1] {$s_5$\\$p$};
      \node [] (5) [rectangle, below=of 2] {$s_4$\\$\neg p$};
      \node [] (6) [rectangle, below=of 3] {$s_6$\\$\neg p$};
      \path 
      (0) edge[->]  node {} (1)
      (0) edge[->]  node {} (2)
      (0) edge[->]  node {} (3)
      (1) edge[->]  node {} (4)
      (2) edge[->]  node {} (5)
      (3) edge[->]  node {} (6)
      (4) edge[->, loop below]  node {} (4)
      (5) edge[->, loop below]  node {} (5)
      (6) edge[->, loop below]  node {} (6);
      \path (4)--node[sloped]{$\eqstate{o_1,\textcolor{red}{o_2}}$}(5);
      \path (1)--node[sloped]{$\eqstate{o_1}$}(2);
      \path (1)--node[sloped]{$\eqstate{o_2}$}(3);
      \path (4)--node[sloped]{$\eqstate{o_2,o_1}$}(6);
    \end{tikzpicture}
    \caption{Model $M'$ in the proof of Proposition~\ref{prop-mono-more-expressive}}
    \label{fig:proof}
\end{figure}
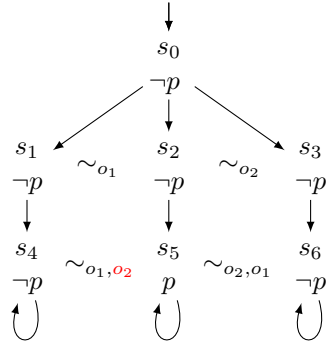

\begin{proposition}
  \label{prop-mono-more-expressive}
  If $|\Obs|>1$ then $\ctlskd\not\lexpr\ctlsk$.
\end{proposition}
\begin{proof}
Assume that $\Obs$ contains $o_1$ and $o_2$.  Consider the model
$M$ from Example~\ref{ex-security} (Figure~\ref{fig:ex}), and define
the model $M'$ which is the same as $M$ except that $s_4$ and $s_5$
are indistinguishable for both $o_1$ and $o_2$, while in $M$ they are
only indistinguishable for $o_1$ (see Figure~\ref{fig:proof}). In both
models, agent $a$ is initially assigned observation $o_1$. 
We exhibit a formula of \ctlskd that can
distinguish between $M$ and $M'$, and justify that no formula of
\ctlsk can, which shows that $\ctlskd\not\lexpr\ctlsk$.

  Consider  formula $\phi=\E F \D{o_2} \K_a p$. With a reasoning similar
to that  detailed in
  Example~\ref{ex-security},
  we can show  that  $M\models\phi$. We now show that
  $M'\not\models\phi$:  The only history in which agent $a$ may
  get to know that $p$ holds is the path $s_0s_2s_5$, since in other
  histories $p$ does not hold. After observing this path with
  observation $o_1$, agent $a$  considers that both $s_4$ and $s_5$ are
  possible. She still does after switching to observation $o_2$, as
  $s_4$ and $s_5$ are $o_2$-indistinguishable. As a result
  $M'\not\models\phi$, and thus $\phi$ distinguishes $M$ and $M'$.

  Now to see that no formula of \ctlsk can distinguish between these two
  models, it is enough to see that in both models the only agent $a$
  is assigned observation $o_1$, and thus on these models no operator of \ctlsk can
  refer to observation $o_2$, which is the only difference between $M$
  and $M'$.  
\end{proof}

  This proof for the mono-agent case relies on the fact that \ctlskd
  can refer to observations that are not initially assigned to any
  agent, and thus cannot be referred to within \ctlsk. The same proof
  can be easily adapted to the multi-agent case, by considering the
  same models $M$ and $M'$ and assigning the same initial observation
  $o_1$ to all agents. We show that in fact, when we have at least two
  agents, \ctlskdm is strictly more expressive than
  \ctlskm even when we only consider  models in which all observations are
  initially assigned to some agent. 

  \begin{proposition}
  \label{prop-poly-more-expressive}
  If $|\Obs|>1$ and $m\geq 2$, it holds that $\ctlskdm\not\lexpr\ctlskm$ even when
  restricted to models where all observations are initially assigned.
\end{proposition}

\begin{proof}
Assume that $\Obs$ contains $o_1$ and $o_2$. We consider the case of
two agents $a$ and $b$ ($m=2$); the proof can easily be generalised to more
agents. Consider again the models
$M$ and $M'$ used in the proof of
Proposition~\ref{prop-mono-more-expressive}. This time, in both
models, agent $a$ is initially assigned observation $o_1$  and agent $b$
 observation $o_2$.
For the same reasons as before,  formula $\phi=\E F \D{o_2} \K_a p$
distinguishes between $M$ and $M'$.

  Now to see that no formula of \ctlskm can distinguish these two
  models, recall that the only difference between $M$ and $M'$
  concerns observation $o_2$, and that agents $a$ and $b$ are bound to
  observations $o_1$ and $o_2$ respectively. Since in \ctlskm agents
  cannot change observation, the modification of $o_2$ between $M$ and
  $M'$ can only affect the knowledge of agent $b$, by letting her
  consider in $M'$ that history $s_0s_1s_4$ is indistinguishable to
  history $s_0s_2s_5$. However this is not the case: indeed, $s_1$
  and $s_2$ are not $o_2$-indistinguishable, and because we consider perfect
  recall, $s_0s_1s_4$ and $s_0s_2s_5$ are not $o_2$-indistinguishable
  neither, even though $s_4$ and $s_5$ are.

  More formally, define the \emph{perfect-recall unfolding}  of a model $M$ as the infinite tree
  consisting of all possible histories starting in the initial state,
  in which two nodes $h$ and $h'$ are related for $o_i$ if $|h|=|h'|$
  and for all $i<|h|$, $h_i\eqstate{o_i}h'_i$. By its semantics it is
  clear that
  $\ctlskm$ is invariant under perfect-recall unfolding. Now it
  suffices to notice that the perfect-recall unfoldings of $M$ and $M'$  are the same, and thus
  cannot be distinguished by any \ctlskm formula.
\end{proof}

\begin{remark}
  \label{rem-unfoldings}
  Unlike \ctlskm,  \ctlskdm is not invariant under the perfect-recall
  unfoldings considered in the proof of
  Proposition~\ref{prop-poly-more-expressive}. The reason is that
  in these unfoldings, observation relations on histories are defined
  for fixed observations, and therefore cannot account for 
  observation changes induced by operators $\Da{o}$.
\end{remark}

Putting together
Propositions~\ref{prop-more-expressive},~\ref{prop-mono-more-expressive}
and~\ref{prop-poly-more-expressive}, we obtain:

\begin{theorem}
  \label{theo-expressive-power}
If $|\Obs|>1$ then $\ctlskm \slexpr \ctlskdm$.
\end{theorem}



\section{Eliminating observation change}
\label{sec:reduction}
In this section we show how to reduce the model-checking problem for
\ctlskd to that of \ctlsk. The approach can be easily generalised to
the multi-agent case.

Fix an instance $(M,\Phi)$ of the model-checking problem for \ctlskd,
where $M$ is a
(mono-agent) model and $\Phi$ is a \ctlskd formula. We build an
equivalent instance $(M',\Phi')$ of the model-checking problem for
\ctlsk; in particular, $M'$ contains a single observation relation,
and $\Phi'$ does not use the observation-change operator~$\D{o}$.

Assume that
$M=(\APf,S,T,V,\{\eqstate{o}\}_{o\in\Obs},\sinit,\oinit)$. We first define the model $M'$. For each observation symbol
$o\in\Obs$ we create a copy $M_o$ of the
original model $M$. Moving to copy $M_o$ will simulate switching to
observation $o$. To make this possible, we need to introduce
transitions between each state $s_o$ of a copy $M_o$ to state
$s_{o'}$ of copy $M_{o'}$, for all $o\neq o'$.  

Let $M'=(\APf\union\{p_o\mid
o\in\Obs\},S',T',V',\eqstate{}',\sinit')$, where
\begin{itemize}
\item for each $o\in\Obs$, $p_o$ is a fresh atomic proposition,
\item $S'=\bigunion_{o\in\Obs}\{s_o\mid s\in S\}$,
\item $T'=\{(s_o,s'_o)\mid o\in \Obs \text{ and }(s,s')\in T\}$

  \hspace{25pt}$\union\;
  \{(s_o,s_{o'})\mid s\in S, o,o'\in\Obs \text{ and }o\neq o'\}$
\item $V'(s_o)=V(s)\union\{p_o\}$, for all $s\in S$ and $o\in\Obs$,
\item $\eqstate{}'= \bigunion_{o\in\Obs}\{(s_o,s'_o) \mid s
  \eqstate{o} s'\}$, and
  \item $\sinit'=\sinit_{\oinit}$.
\end{itemize}

We now define formula $\Phi'$. The
transformation $\tro$ is parameterised with an observation $o\in\Obs$
and is defined by induction on $\Phi$ as follows:
\begin{align*}
  \tro(\D{o'}\phi) &=
                     \begin{cases}
                       \tro[o'](\phi) & \text{if }o=o'\\
                       \A\X(p_{o'}\rightarrow \tro[o'](\phi)) & \text{otherwise}
                     \end{cases}\\
  \tro(\A\psi) & = \A (\G p_o \impl \tro(\psi))
\end{align*}

All other cases simply distribute over operators. We finally let $\Phi'=\tro[\oinit](\Phi)$.

Using the alternative semantics it is rather easy to see that the
following holds:

\begin{lemma}
  \label{lem-reduction}
  $M\models\Phi$ if, and only if, $M'\models\Phi'$.
\end{lemma}

This establishes the correctness of the reduction, and because we know
how to model-check \ctlsk, provides a model-checking procedure for
\ctlskd. However this algorithm does not provide optimal
complexity. Indeed, the model $M'$ is of size $|M|\times|\Obs|$, and the
best known model-checking algorithm for \ctlsk runs in time
exponential in the size of the model and the size of the
formula~\cite{BOZZELLI201580}. Going through this reduction thus
yields a procedure that is
exponential in the number of observations. Our direct model-checking
procedure, which generalises techniques used for the classic case of
static observations, provides instead a decision procedure which is
only polynomial in the number of observations (Theorem~\ref{theo-compl}).

The reduction described above can be easily generalised to the
multi-agent case, by creating one copy $M_{\ovec}$ of the original model
$M$ for each possible assignment  $\ovec$ of observations to agents. We
thus get a model $M'$ of size $|M|\times{|\Obs|^{|\Ag|}}$, and since the best
known model-checking procedure for \ctlskm is $k$-exponential in the
size of the model~\cite{BOZZELLI201580}, this reduction provides a procedure
which is $k$-exponential in the number of observations and
$k+1$-exponential in the number of agents. The direct approach
provides an algorithm that is only polynomial in the number of
observations, exponential in the number of agents, and whose combined
complexity is $k$-exponential time (Theorem~\ref{theo-complex-multi}).


\section{Conclusion and future work}
\label{sec:conclusion}
Epistemic temporal logics play a central role in MAS 
as they permit to describe in an elegant  way the
knowledge of agents along the evolution of a system.
Previous works in this field have  treated agents' observation power as a static feature. 
However, in many real life scenarios,  agents'
 observation power may change along computations. 

In this work we  introduced \ctlskd, a logic that 
can express such dynamic changes of  observation power, and  we
demonstrated how this can be used to
 express relevant properties in practical scenarios.
We  studied the model checking of \ctlskd over both mono- and multi-agent
systems and proved that in both cases, the ability to express observation
changes comes at no complexity cost, but strictly increases
expressivity. We also showed how to reduce the model-checking problem
for our logic to that of \ctlsk, removing the observation-change
operator. This provides a model-checking procedure for \ctlskd, but
given the complexity of the best-known model-checking algorithm for \ctlsk,
this procedure is not as efficient as the direct algorithm we provide.   

As future work, we plan to 
build upon the techniques developed here to investigate epistemic
extensions of strategic
logics with imperfect information.  Several such logics
have been defined and studied
recently~\cite{epistemicSL,BerthonMM17a,DBLP:conf/lics/BerthonMMRV17,BelardinelliLMR17b,BelardinelliLMR17a},
and~\cite{DBLP:conf/lics/BerthonMMRV17} in particular already presents the
feature of dynamic observation change via change of strategy. We believe
that the present work will help to establish new results on the model
checking of such logics.




\begin{thebibliography}{10}

\bibitem{DBLP:conf/atal/Aucher14}
Guillaume Aucher.
\newblock Supervisory control theory in epistemic temporal logic.
\newblock In {\em {AAMAS}}, pages 333--340, 2014.
\newblock URL: \url{http://dl.acm.org/citation.cfm?id=2615787}.

\bibitem{epistemicSL}
Francesco Belardinelli.
\newblock A logic of knowledge and strategies with imperfect information.
\newblock In {\em LAMAS}, pages 1--15, 2015.

\bibitem{BelardinelliLMR17b}
Francesco Belardinelli, Alessio Lomuscio, Aniello Murano, and Sasha Rubin.
\newblock Verification of broadcasting multi-agent systems against an epistemic
  strategy logic.
\newblock In {\em {IJCAI}}, pages 91--97, 2017.

\bibitem{BelardinelliLMR17a}
Francesco Belardinelli, Alessio Lomuscio, Aniello Murano, and Sasha Rubin.
\newblock Verification of multi-agent systems with imperfect information and
  public actions.
\newblock In {\em {AAMAS}}, pages 1268--1276, 2017.

\bibitem{BerthonMM17a}
Rapha{\"{e}}l Berthon, Bastien Maubert, and Aniello Murano.
\newblock Decidability results for {ATL}* with imperfect information and
  perfect recall.
\newblock In {\em {AAMAS}}, pages 1250--1258, 2017.

\bibitem{DBLP:conf/lics/BerthonMMRV17}
Rapha{\"{e}}l Berthon, Bastien Maubert, Aniello Murano, Sasha Rubin, and
  Moshe~Y. Vardi.
\newblock Strategy logic with imperfect information.
\newblock In {\em {LICS}}, pages 1--12, 2017.
\newblock URL: \url{https://doi.org/10.1109/LICS.2017.8005136}, \href
  {http://dx.doi.org/10.1109/LICS.2017.8005136}
  {\path{doi:10.1109/LICS.2017.8005136}}.

\bibitem{bittner2012symbolic}
Benjamin Bittner, Marco Bozzano, Alessandro Cimatti, and Xavier Olive.
\newblock Symbolic synthesis of observability requirements for diagnosability.
\newblock In {\em AAAI}, 2012.

\bibitem{BOZZELLI201580}
Laura Bozzelli, Bastien Maubert, and Sophie Pinchinat.
\newblock Uniform strategies, rational relations and jumping automata.
\newblock {\em Information and Computation}, 242:80 -- 107, 2015.
\newblock URL:
  \url{http://www.sciencedirect.com/science/article/pii/S0890540115000279},
  \href {http://dx.doi.org/https://doi.org/10.1016/j.ic.2015.03.012}
  {\path{doi:https://doi.org/10.1016/j.ic.2015.03.012}}.

\bibitem{DBLP:conf/fossacs/BozzelliMP15}
Laura Bozzelli, Bastien Maubert, and Sophie Pinchinat.
\newblock Unifying hyper and epistemic temporal logics.
\newblock In {\em FoSSaCS}, pages 167--182, 2015.
\newblock URL: \url{https://doi.org/10.1007/978-3-662-46678-0_11}, \href
  {http://dx.doi.org/10.1007/978-3-662-46678-0_11}
  {\path{doi:10.1007/978-3-662-46678-0_11}}.

\bibitem{Dima2009}
C{\u{a}}t{\u{a}}lin Dima.
\newblock Revisiting satisfiability and model-checking for ctlk with synchrony
  and perfect recall.
\newblock In {\em CLIMA IX-2008}, pages 117--131, 2009.
\newblock URL: \url{https://doi.org/10.1007/978-3-642-02734-5_8}, \href
  {http://dx.doi.org/10.1007/978-3-642-02734-5_8}
  {\path{doi:10.1007/978-3-642-02734-5_8}}.

\bibitem{emerson1987modalities}
E~Allen Emerson and Chin-Laung Lei.
\newblock Modalities for model checking: Branching time logic strikes back.
\newblock {\em Science of computer programming}, 8(3):275--306, 1987.

\bibitem{RAK}
Ronald Fagin, Joseph~Y Halpern, Yoram Moses, and Moshe Vardi.
\newblock {\em Reasoning about knowledge}.
\newblock MIT press, 2004.

\bibitem{DBLP:journals/jcs/HalpernO05}
Joseph~Y. Halpern and Kevin~R. O'Neill.
\newblock Anonymity and information hiding in multiagent systems.
\newblock {\em Journal of Computer Security}, 13(3):483--512, 2005.
\newblock URL:
  \url{http://content.iospress.com/articles/journal-of-computer-security/jcs237}.

\bibitem{DBLP:journals/siamcomp/HalpernMV04}
Joseph~Y. Halpern, Ron van~der Meyden, and Moshe~Y. Vardi.
\newblock Complete axiomatizations for reasoning about knowledge and time.
\newblock {\em {SIAM} J. Comput.}, 33(3):674--703, 2004.
\newblock URL: \url{https://doi.org/10.1137/S0097539797320906}, \href
  {http://dx.doi.org/10.1137/S0097539797320906}
  {\path{doi:10.1137/S0097539797320906}}.

\bibitem{DBLP:journals/logcom/JamrogaT18}
Wojciech Jamroga and Masoud Tabatabaei.
\newblock Accumulative knowledge under bounded resources.
\newblock {\em J. Log. Comput.}, 28(3):581--604, 2018.
\newblock URL: \url{https://doi.org/10.1093/logcom/exv003}, \href
  {http://dx.doi.org/10.1093/logcom/exv003} {\path{doi:10.1093/logcom/exv003}}.

\bibitem{DBLP:conf/atal/KongL17}
Jeremy Kong and Alessio Lomuscio.
\newblock Symbolic model checking multi-agent systems against {CTL*K}
  specifications.
\newblock In {\em AAMAS}, pages 114--122, 2017.
\newblock URL: \url{http://dl.acm.org/citation.cfm?id=3091147}.

\bibitem{kupferman2000automata}
Orna Kupferman, Moshe~Y Vardi, and Pierre Wolper.
\newblock An automata-theoretic approach to branching-time model checking.
\newblock {\em Journal of the ACM (JACM)}, 47(2):312--360, 2000.

\bibitem{DBLP:conf/tark/LadnerR86}
Richard~E. Ladner and John~H. Reif.
\newblock The logic of distributed protocols.
\newblock In {\em TARK}, pages 207--222, 1986.

\bibitem{DBLP:journals/tocl/MogaveroMPV14}
Fabio Mogavero, Aniello Murano, Giuseppe Perelli, and Moshe~Y. Vardi.
\newblock Reasoning about strategies: On the model-checking problem.
\newblock {\em {ACM} Trans. Comput. Log.}, 15(4):34:1--34:47, 2014.
\newblock URL: \url{http://doi.acm.org/10.1145/2631917}, \href
  {http://dx.doi.org/10.1145/2631917} {\path{doi:10.1145/2631917}}.

\bibitem{pacuit2007some}
Eric Pacuit.
\newblock Some comments on history based structures.
\newblock {\em Journal of Applied Logic}, 5(4):613--624, 2007.

\bibitem{raimondi2005complexity}
Franco Raimondi and Alessio Lomuscio.
\newblock The complexity of symbolic model checking temporal-epistemic logics.
\newblock In {\em CS\&P}, pages 421--432, 2005.

\bibitem{sampath1995diagnosability}
Meera Sampath, Raja Sengupta, St{\'e}phane Lafortune, Kasim Sinnamohideen, and
  Demosthenis Teneketzis.
\newblock Diagnosability of discrete-event systems.
\newblock {\em IEEE Transactions on automatic control}, 40(9):1555--1575, 1995.

\bibitem{van2003cooperation}
W.~van~der Hoek and M.~Wooldridge.
\newblock {Cooperation, knowledge, and time: {A}lternating-time {T}emporal
  {E}pistemic {L}ogic and its applications}.
\newblock {\em Studia Logica}, 75(1):125--157, 2003.
\newblock \href {http://dx.doi.org/10.1023/A:1026185103185}
  {\path{doi:10.1023/A:1026185103185}}.

\bibitem{DBLP:journals/iandc/Meyden98}
Ron van~der Meyden.
\newblock Common knowledge and update in finite environments.
\newblock {\em Inf. Comput.}, 140(2):115--157, 1998.
\newblock URL: \url{https://doi.org/10.1006/inco.1997.2679}, \href
  {http://dx.doi.org/10.1006/inco.1997.2679}
  {\path{doi:10.1006/inco.1997.2679}}.

\bibitem{DBLP:conf/fsttcs/MeydenS99}
Ron van~der Meyden and Nikolay~V. Shilov.
\newblock Model checking knowledge and time in systems with perfect recall
  (extended abstract).
\newblock In {\em FSTTCS}, pages 432--445, 1999.

\bibitem{DBLP:conf/csfw/MeydenS04}
Ron van~der Meyden and Kaile Su.
\newblock Symbolic model checking the knowledge of the dining cryptographers.
\newblock In {\em {CSFW-17}}, pages 280--291, 2004.

\bibitem{van1998synthesis}
Ron van~der Meyden and Moshe~Y Vardi.
\newblock Synthesis from knowledge-based specifications.
\newblock In {\em CONCUR}, pages 34--49. Springer, 1998.

\bibitem{ditmarsch2007dynamic}
Hans van Ditmarsch, Wiebe Van~der Hoek, and Barteld~Pieter Kooi.
\newblock {\em Dynamic epistemic logic}, volume 337.
\newblock Springer, 2007.

\bibitem{von2007theory}
John Von~Neumann and Oskar Morgenstern.
\newblock {\em Theory of games and economic behavior (commemorative edition)}.
\newblock Princeton university press, 2007.

\end{thebibliography}

\newpage
\appendix

\section{Proof of Theorem~\ref{theo-equivalence}}
\label{sec-app}

Theorem~\ref{theo-equivalence} directly follows from the following lemma:
\begin{lemma}
  \label{theo-extended}
  For every history formula $\phi$, model $M$, history $h$ and
  observation record $r$ that stops at $h$,
  \[h,r\models\phi \quad\mbox{iff}\quad
  \last(h),\IS(h,r),O(h,r)\modelsI \phi\] and
  for every path formula $\psi$, every path $\pi$,   point in time
  $n\in\setn$ and observation record
  $r$ that stops at $n$,
  \[\pi,n,r\models\psi\quad\mbox{iff}\quad
  \pi_{\geq n},I(\pi_{\leq n},r),O(\pi_{\leq n},r)\modelsI\psi.\]
\end{lemma}

\def\iff{\quad\mbox{iff}\quad}

\begin{proof} We prove the result by mutual induction on $\phi$ and $\psi$.
Let $M$ be a model, $h$ a history and $r$ an observation record that
stops at $h$.

$\bm{\phi=p}$: We have
$h,r\models p\iff p\in V(\last(h))\iff
\last(h),\IS(h,r),O(h,r)\modelsI p$, by applying the definitions.

\vspace{1ex}
$\bm{\phi=\neg\varphi'}$ and $\bm{\phi=\phi_1\wedge\phi_2}$: Simple
application of the
induction hypothesis.

\vspace{1ex} $\bm{\phi=\A\psi}$: We start with the left-to-right
implication. Assume that $h,r\models\A\psi$.  We
need to prove that for all
$\pi$ such that $\pi_0=\last(h)$,
$\pi,I(h,r),O(h,r)\modelsI\psi$.  If $\pi_0=\last(h)$, we have $h\pref
h\cdot\pi_{>0}$; since $h,r\models\A\psi$ we thus
get that $h\cdot\pi_{>0},|h|-1,r\models\psi$.  By induction hypothesis
we have $(h\cdot\pi_{>0})_{>|h|-1},I((h\cdot\pi_{>0})_{\leq |h|-1}, r),
O((h\cdot\pi_{>0})_{\leq |h|-1},r)\modelsI\psi$, that is to say
$\pi,I(h,r),O(h,r)\modelsI\psi$, as $\pi_0=\last(h)$.

For the other direction, assume that
$\last(h),I(h,r),O(h,r)\modelsI\A\psi$, \ie,  for all $\pi$ such that
$\pi_0=\last(h)$, it holds that $\pi,I(h,r),O(h,r)\modelsI\psi$.
  Let $\pi$ be such that $h\pref\pi$; we prove that $\pi,|h|-1,r\models\psi$.
  Because $h\pref\pi$, we have that $(\pi_{\geq |h|-1})_0 =
  \last(h)$. By assumption, it follows that $\pi_{\geq |h|-1},I(h,r),O(h,r)\modelsI\psi$,
  that is $\pi_{\geq |h|-1},I(\pi_{\leq|h|-1},r),O(\pi_{\leq|h|-1},r)\modelsI\psi$.
  Then by induction hypothesis we finally get that $\pi,|h|-1,r\models\psi$.

  \vspace{1ex}
  $\bm{\phi=\K\phi':}$ For the first direction, assume that
  $h,r\models\K\phi'$, \ie, for all
  $h'\eqh{r}h$, we have $h',r\models\phi$.
We prove that $\last(h),I(h,r),O(h,r)\modelsI\K\phi'$, \ie, for all
$s'\in I(h,r)$, $s',I(h,r),O(h,r)\modelsI\phi'$.
Let $s'\in I(h,r)$. By definition of the information set, there exists $h'\eqh{r}h$ such that $h'\eqh{r}h$ and $\last(h')=s'$.
By assumption, $h',r\models\phi$, and by induction hypothesis, $\last(h'),I(h',r),O(h',r)\modelsI\phi$.
  Because $h'\eqh{r}h$, $I(h,r)=I(h',r)$ and $O(h,r)=O(h',r)$. It
  follows that $s',I(h,r),O(h,r)\modelsI\phi$.

For the other direction,  we assume that for all $s'\in I(h,r),
s',I(h,r),O(h,r)\modelsI\phi'$, and we prove that $h,r\models\K\phi'$,
\ie, for all $h'\eqh{r}h$, $h',r\models\phi'$.
  Let $h'$ be such that $h'\eqh{r}h$. By definition, $\last(h')\in I(h,r)$.
  By assymption, $\last(h'),I(h,r),O(h,r)\modelsI\phi'$. Because
  $h'\eqh{r}h$, $I(h,r)=I(h',r)$ and $O(h,r)=O(h',r)$, and we get that
  $\last(h'),I(h',r),O(h',r)\modelsI\phi'$. By induction hypothesis, $h',r\models\phi'$.

  \vspace{1ex} $\bm{\phi=\D{o}\phi'}$: By definition,
  $h,r\models\D{o}\phi\iff h,r.(o,|h|-1)\models\phi$. By induction
  hypothesis,
  $h,r.(o,|h|-1)\models\phi\iff
  \last(h),I(h,r.(o,|h|-1)),O(h,r.(o,|h|-1))\modelsI\phi$.  By
  Proposition~\ref{lem-UT}, because $r$ stops at $|h|-1$, this is
  equivalent to $\last(h),\UD(I(h,r),o),o\modelsI\phi$, which by
  definition is equivalent to $\last(h),I(h,r),O(h,r)\modelsI\D{o}\phi$.

  \vspace{1ex}
Let $\pi$ be a path, $n$ a natural number and $r$ an observation
record that stops at $n$.

\vspace{1ex}
$\bm{\psi=\phi}$: By definition, $\pi,n,r\models\phi\iff \pi_{\leq
  n},r\models\phi$. By induction hypothesis, the latter is equivalent
to $\last(\pi_{\leq n}),I(\pi_{\leq n},r),O(\pi_{\leq
  n},r)\modelsI\phi$. Because $(\pi_{\geq n})_0=\last(\pi_{\leq n})$,
by definition this is also equivalent to
 $\pi_{\geq n},I(\pi_{\leq n},r),O(\pi_{\leq n},r)\modelsI\phi$, which concludes.

  \vspace{1ex}
  $\bm{\psi=\neg\psi'}$ and $\bm{\psi=\psi_1\wedge\psi_2}$:
  By application of the induction hypothesis.

  \vspace{1ex} $\bm{\psi=\X\psi'}$: By definition,
  $\pi,n,r\models\X\psi'\iff \pi,n+1,r\models\psi'$. By induction
  hypothesis we get
  \begin{equation}
    \label{eq:a}
\pi,n,r\models\X\psi' \iff  \pi_{\geq(n+1)},I(\pi_{\leq(n+1)},r),O(\pi_{\leq(n+1)},r)\modelsI\psi'.    
  \end{equation}
  Because $r$ stops at $n$, we have that
  $O(\pi_{\leq(n+1)},r)=O(\pi_{\leq n},r)$.  Using
  Proposition~\ref{lem-UT}, we also have that
  $I(\pi_{\leq(n+1)},r)=\UT(I(\pi_{\leq n},r),\pi_{n+1},O(\pi_{\leq
    n},r))$. \eqref{eq:a} thus becomes
  \begin{equation}
    \label{eq:b}
\pi,n,r\models\X\psi' \iff
\pi_{\geq(n+1)},\UT(I(\pi_{\leq n},r),\pi_{n+1},O(\pi_{\leq
    n},r)),O(\pi_{\leq n},r)\modelsI\psi'.    
  \end{equation}
Because  $\pi_{(n+1)} = (\pi_{\geq n})_{1}$, by definition we finally
get that
  \[\pi,n,r\models\X\psi\iff \pi_{\geq n},I(\pi_{\leq n},r),O(\pi_{\leq
    n},r)\modelsI\X\psi'.\]

  \vspace{1ex}
  $\bm{\psi=\psi_1\U\psi_2}$: According to the
  definitions, it suffices to prove that for all path  $k\geq 0$ and $\psi\in\{\psi_1,\psi_2\}$,
  $\pi,n+k,r\models\psi\iff \pi_{\geq (n+k)},\UT^k(I(\pi_{\leq
    n},r),\pi_{\geq n},O(\pi_{\leq n},r)),O(\pi_{\leq
    n},r)\modelsI\psi$.  This can be proven by induction on $k$, using
for the inductive step  a reasoning similar to that for
$\psi=\X\psi'$.
  \end{proof}

\section{Proof of Proposition~\ref{lem-UTk}}  
\label{sec-app-b}

\setcounter{theorem}{\value{lem-UTk}}
\begin{proposition}
For every history $h\cdot s$, record tuple $\rvec$ that stops at $h$, observation tuple $\ovec$  and 
integer $k$, it holds that
\[
  \begin{array}{rcl}
  \KT{k}(h\cdot s,\rvec)&=&\UTK{k}(\KT{k}(h,\rvec),s,\ovec(h,\rvec)),\mbox{ and}\\
    \KT{k}(h,\append[\rvec]{(o,|h|-1)_a})&=&\UDK{k}(\KT{k}(h,\rvec),o,a).
  \end{array}
\]
\end{proposition}

\begin{proof}
  The proof is by induction on $k$.

  $\bm{k=0}$: the result is immediate
from the definitions.

$\bm{k+1}$: We start with the first part of the proposition.

By definition $\KT{k+1}(h\cdot s,\rvec)=\langle
  s,\forest_1,\ldots,\forest_m\rangle$,
where for each $i$, \[\forest_i=\{\KT{k}(h'\cdot s',\rvec) \mid h'\cdot s'
  \eqha[a_i]{\rvec} h\cdot s\}.\]
On the other hand, letting $\KT{k+1}(h,\rvec)=\langle
\last(h),\forest'_1,\ldots,\forest'_m\rangle$ and
$\ovec=\ovec(h,\rvec)$, we have by definition
that
$\UTK{k+1}(\KT{k+1}(h,\rvec),s,\ovec)=\langle s, \forest''_1,
                                  \ldots, \forest''_m\rangle$,
where for each $i$, \[\forest''_i\egdef \{\UTK{k}(\KT{k},s',\ovec)\mid
\KT{k}\in \forest'_i\mbox{, } s' \eqstate{\ovec_i} s \mbox{ and
}\racine(\KT{k})T s'\}.\]
It thus suffices to show that for each $i$, $\forest_i=\forest''_i$.
Observe that because $\rvec$ stops at $h$, $h'\cdot s'
  \eqha[a_i]{\rvec} h\cdot s$ iff $h'\eqha[a_i]{\rvec} h$ and
  $s'\eqstate{\ovec_i}s$. We thus have that
  \[\forest_i=\{\KT{k}(h'\cdot s',\rvec) \mid h' \eqha[a_i]{\rvec} h,\;s' \eqstate{\ovec_i} s \mbox{ and
    }\last(h')T s'\}.\]
  By induction hypothesis, this becomes
    \[\forest_i=\{\UTK{k}(\KT{k}(h',\rvec),s',\ovec(h',\rvec)) \mid h' \eqha[a_i]{\rvec} h,\;s' \eqstate{\ovec_i} s \mbox{ and
      }\last(h')T s'\}.\]
    For each $h' \eqha[a_i]{\rvec} h$, we have that
    $\ovec(h',\rvec)=\ovec(h,\rvec)=\ovec$.

    It remains to observe that by definition, for each $i$, $\forest'_i=\{\KT{k}(h',\rvec) \mid h'
  \eqha[a_i]{\rvec} h\}$, and we get that $\forest_i=\forest''_i$,
  which concludes the first part of the proposition. 

  \vspace{1ex}
  We now prove the second part of the proposition, assuming that $a=a_j$.

  By definition 
$\KT{k+1}(h,\append[\rvec]{(o,|h|-1)_{a_j}})=\langle
  \last(h),\forest_1,\ldots,\forest_m\rangle$,
where for each $i$, \[\forest_i=\{\KT{k}(h',\append[\rvec]{(o,|h|-1)_{a_j}}) \mid h'
  \eqha[a_i]{\append[\rvec]{(o,|h|-1)_{a_j}}} h\}.\]

On the other hand, letting
$\KT{k+1}(h,\rvec)=\langle
\last(h),\forest'_1,\ldots,\forest'_m\rangle$, we have by definition that
$\UDK{k+1}(\KT{k+1}(h,\rvec),o,a_j) \egdef
\langle \last(h), \forest''_1, \ldots,\forest''_m\rangle$, where for each
$i\neq j$,
\[\forest''_i=\{\UDK{k}(\KT{k},o,a_j)\mid \KT{k}\in\forest'_i\}, \mbox{
    and}\]
\[\forest''_j\egdef \{\UDK{k}(\KT{k},o,a_j)\mid
\KT{k}\in\forest'_j \mbox{ and } \racine(\KT{k})\eqstate{o} \last(h)\}.\]
We now show that for all $i$, $\forest_i=\forest''_i$.

First, by induction hypothesis, for each $i$ we have 
 \[\forest_i=\{\UDK{k}(\KT{k}(h',\rvec),o,a_j)\mid h'
  \eqha[a_i]{\append[\rvec]{(o,|h|-1)_{a_j}}} h\}.\]

Also, by definition, for each $i$ we have $\forest'_i=\{\KT{k}(h',\rvec) \mid h'
  \eqha[a_i]{\rvec} h\}$.

Next, for $i\neq j$ we have that $h'
  \eqha[a_i]{\append[\rvec]{(o,|h|-1)_{a_j}}} h$ iff $h'
  \eqha[a_i]{\rvec} h$, and by definition of the perfect-recall
  relation, for $i=j$, it holds that $h'
  \eqha[a_i]{\append[\rvec]{(o,|h|-1)_{a_j}}} h$ iff $h'
  \eqha[a_i]{\rvec} h$ and $\last(h')\eqstate{o}\last(h)$. 
It is then easy to see that for all $i$, $\forest_i=\forest''_i$,
which concludes.
\end{proof}


\end{document}